%% file: main.tex
\def\llncs{0}
\def\fullpage{1}
\def\anonymous{0}
\def\authnote{1}
\def\notxfont{0}
\def\submission{0}
\def\cameraready{0}

\ifnum\submission=1
\def\anonymous{1}
\def\llncs{1}
\def\authnote{0}
\else
\fi

\ifnum\cameraready=1
\def\submission{1}
\def\llncs{1}
\def\authnote{0}
\def\anonymous{0}
\else
\fi

\ifnum\anonymous=1
\def\llncs{1}
\def\authnote{0}
\else
\fi

\ifnum\llncs=1
	\documentclass[envcountsect,a4paper,runningheads,10pt]{llncs}
\else
	\documentclass[letterpaper,hmargin=1.05in,vmargin=1.05in]{article}
			\ifnum\fullpage=1
		\usepackage{fullpage}
		\fi
\fi

\input{preamble_usepackages}
\input{macros}

\ifnum\llncs=1
\let\oldvec\vec
\let\vec\oldvec
\makeatletter
\renewcommand*\l@author[2]{}
\renewcommand*\l@title[2]{}
\makeatletter

\else
\fi    

\theoremstyle{remark}

\title{
\textbf{Certified Everlasting Zero-Knowledge Proof for QMA}
}

\begin{document}

\ifnum\anonymous=1
\author{\empty}\institute{\empty}
\else
%
%
\ifnum\llncs=1
\author{
	Taiga Hiroka\inst{1} \and Tomoyuki Morimae\inst{1,2} \and Ryo Nishimaki\inst{3} \and Takashi Yamakawa\inst{3}
}
\institute{
	Yukawa Institute for Theoretical Physics, Kyoto University, Japan \and PRESTO, JST, Japan \and NTT Secure Platform Laboratories
}
\else
%
%

\author[1]{Taiga Hiroka}
\author[1]{\hskip 1em Tomoyuki Morimae}
\author[2]{\hskip 1em Ryo Nishimaki}
\author[2]{\hskip 1em Takashi Yamakawa}
\affil[1]{{\small Yukawa Institute for Theoretical Physics, Kyoto University, Kyoto, Japan}\authorcr{\small \{taiga.hiroka,tomoyuki.morimae\}@yukawa.kyoto-u.ac.jp}}
\affil[2]{{\small NTT Corporation, Tokyo, Japan}\authorcr{\small \{ryo.nishimaki.zk,takashi.yamakawa.ga\}@hco.ntt.co.jp}}
\renewcommand\Authands{, }
\fi 
\fi

\ifnum\llncs=1
\date{}
\else
\date{\today}
\fi

\maketitle

\begin{abstract}
In known constructions of classical zero-knowledge protocols for $\NP$, either of zero-knowledge or soundness holds only 
against computationally bounded adversaries.
Indeed, achieving both statistical zero-knowledge and statistical soundness at the same time with classical verifier 
is impossible for $\NP$ unless the polynomial-time hierarchy collapses,
and it is also believed to be impossible even with a quantum verifier.
In this work, we introduce a novel compromise, which we call
the certified everlasting zero-knowledge proof for $\QMA$.
It is a computational zero-knowledge proof for $\QMA$, but
the verifier issues a classical certificate that shows that the verifier has deleted its quantum information. 
If the certificate is valid, even unbounded malicious verifier can no longer learn anything beyond the validity of the statement.

We construct a certified everlasting zero-knowledge proof for $\QMA$.
For the construction, we introduce a new quantum cryptographic primitive, which we call
commitment with statistical binding and certified everlasting hiding,
where the hiding property becomes statistical once the receiver has issued a valid certificate that
shows that the receiver has deleted the committed information.
We construct commitment with statistical binding and certified everlasting hiding 
from quantum encryption with certified deletion by Broadbent and Islam [TCC 2020] (in a black box way),
and then combine it with the quantum sigma-protocol for $\QMA$ by Broadbent and Grilo [FOCS 2020]
to construct the certified everlasting zero-knowledge proof for $\QMA$.
Our constructions are secure in the quantum random oracle model.
Commitment with statistical binding and certified everlasting hiding itself is of independent interest, and there will be many other 
useful applications beyond zero-knowledge.
\end{abstract}

\ifnum\cameraready=1
\else
\ifnum\llncs=1
\else
\newpage
  \setcounter{tocdepth}{2}      
  \setcounter{secnumdepth}{2}   
  \setcounter{page}{0}          
  \tableofcontents
  \thispagestyle{empty}
  \clearpage
\fi
\fi

\input{Sec_Introduction}

\input{Sec_Preliminaries}

\input{Sec_Commitment}

\input{Sec_Zeroknowledge}

\ifnum\submission=1
\else
\section*{Acknowledgement}
TM is supported by 
the JST Moonshot R\verb|&|D JPMJMS2061-5-1-1,
JST FOREST,
MEXT Q-LEAP, 
and the Grant-in-Aid for Scientific Research (B) No.JP19H04066 of JSPS.
\fi

\input{reference}


\ifnum\cameraready=1
\else
\appendix

	\ifnum\llncs=1
	\newpage
	 	\setcounter{page}{1}
 	{
	\noindent
 	\begin{center}
	{\Large SUPPLEMENTAL MATERIALS}
	\end{center}
 	}
	\setcounter{tocdepth}{2}
	\section{Proof of \texorpdfstring{\cref{thm:computationally_hiding_2}}{Theorem~\ref{thm:computationally_hiding_2}}}\label{Sec:Computational_hiding}

\input{Proof_C_hiding}\input{Sec_Appendix}

	\section{Proof of \texorpdfstring{\cref{thm:sequential_everlasting_zero_proof}}{Theorem~\ref{thm:sequential_everlasting_zero_proof}}}\label{Sec:sequential_everlasting}
	\input{Proof_sequential}
  \input{App_sum_binding}
  \input{proof_of_bit}

	\else

\input{Sec_Appendix}\input{App_sum_binding}\input{proof_of_bit}
	\fi
\fi

\ifnum\cameraready=1
\else
\ifnum\submission=1
\newpage
\setcounter{tocdepth}{1}
\tableofcontents
\else
\fi
\fi

\end{document}

%% file: preamble_usepackages.tex
\usepackage[%
  colorlinks=true,
  citecolor=blue,
  pagebackref=true
]{hyperref}

\usepackage{amsmath, amsfonts, amssymb, mathtools,amscd}

\usepackage{amsthm}

\usepackage{lmodern}
\usepackage[T1]{fontenc}
\usepackage[utf8]{inputenc}

\usepackage{arydshln} 
\usepackage{url}
\usepackage{ifthen}
\usepackage{bm}
\usepackage{multirow}
\usepackage[dvips]{graphicx}
\usepackage[usenames]{color}
\usepackage{xcolor,colortbl} 
\usepackage{threeparttable}
\usepackage{comment}
\usepackage{paralist,verbatim}
\usepackage{cases}
\usepackage{booktabs}
\usepackage{braket}
\usepackage{cancel} 
\usepackage{ascmac} 
\usepackage{framed}
\usepackage{physics}
\usepackage{authblk}
\usepackage{pifont}
\usepackage{autonum}
\definecolor{darkblue}{rgb}{0,0,0.6}
\definecolor{darkgreen}{rgb}{0,0.5,0}
\definecolor{maroon}{rgb}{0.5,0.1,0.1}
\definecolor{dpurple}{rgb}{0.2,0,0.65}

\usepackage[capitalise,noabbrev]{cleveref}
\usepackage[absolute]{textpos}
\usepackage[final]{microtype}
\usepackage[absolute]{textpos}
\usepackage{everypage}

\newtheoremstyle{thicktheorem}%
{\topsep}
{\topsep}
{\itshape}{}%
{\bfseries}%
{.}
{ }%
{\thmname{#1}\thmnumber{ #2}%
		\thmnote{ (#3)}%
}

\newtheoremstyle{remark}
{\topsep}
{\topsep}
	{}
	{}
	{}
	{.}
	{ }
	{\textit{\thmname{#1}}\thmnumber{ #2}
			\thmnote{ (#3)}%
	}

\ifnum\llncs=0
	\theoremstyle{thicktheorem}
	\newtheorem{theorem}{Theorem}[section]
	\newtheorem{lemma}[theorem]{Lemma}
	
	\newtheorem{proposition}[theorem]{Proposition}
	\newtheorem{definition}[theorem]{Definition}
	\newtheorem{game}[theorem]{Game}

	\theoremstyle{remark}
	
	\newtheorem{remark}[theorem]{Remark}

\else
\fi

	\crefname{theorem}{Theorem}{Theorems}
	\crefname{assumption}{Assumption}{Assumptions}
	\crefname{construction}{Construction}{Constructions}
	\crefname{corollary}{Corollary}{Corollaries}
	\crefname{conjecture}{Conjecture}{Conjectures}
	\crefname{definition}{Definition}{Definitions}
	\crefname{exmaple}{Example}{Examples}
	\crefname{experiment}{Experiment}{Experiments}
	\crefname{counterexample}{Counterexample}{Counterexamples}
	\crefname{lemma}{Lemma}{Lemmata}
	\crefname{observation}{Observation}{Observations}
	\crefname{proposition}{Proposition}{Propositions}
	\crefname{remark}{Remark}{Remarks}
	\crefname{claim}{Claim}{Claims}
	\crefname{fact}{Fact}{Facts}
	\crefname{note}{Note}{Notes}

\ifnum\llncs=1
 \crefname{appendix}{App.}{Appendices}
 \crefname{section}{Sec.}{Sections}
\else
\fi

\ifnum\llncs=1
\renewcommand*{\backref}[1]{}
\else
	\renewcommand*{\backref}[1]{(Cited on page~#1.)}
	\ifnum\notxfont=1
	\else
		\usepackage{newtxtext}
	\fi
\fi

%% file: macros.tex
\ifnum\authnote=0  
\newcommand{\mor}[1]{}
\newcommand{\taiga}[1]{}
\newcommand{\ryo}[1]{}
\newcommand{\takashi}[1]{}

\else
\newcommand{\mor}[1]{$\ll$\textsf{\color{red} Tomoyuki: { #1}}$\gg$}
\newcommand{\taiga}[1]{$\ll$\textsf{\color{magenta} Taiga: { #1}}$\gg$}
\newcommand{\takashi}[1]{$\ll$\textsf{\color{orange} Takashi: { #1}}$\gg$}
\newcommand{\ryo}[1]{$\ll$\textsf{\color{darkgreen} Ryo: { #1}}$\gg$}

\fi

\newcommand{\QMA}{\compclass{QMA}}
\newcommand{\hist}{\mathsf{hist}}
\newcommand{\simulator}{\mathsf{sim}}
\newcommand{\statement}{\mathsf{x}}
\newcommand{\yes}{\mathsf{yes}}
\newcommand{\no}{\mathsf{no}}
\newcommand{\witness}{\mathsf{w}}

\newcommand{\Commit}{\algo{Commit}}
\newcommand{\Verify}{\algo{Verify}}

\newcommand{\cment}{\mathsf{com}}
\newcommand{\dment}{\mathsf{d}}

\newcommand{\Classical}{\algo{Classical}}

\newcommand{\Delete}{\algo{Del}}
\newcommand{\Cert}{\algo{Cert}}
\newcommand{\cert}{\mathsf{cert}}

\newcommand{\Sim}{\algo{Sim}}

\newcommand{\NP}{\compclass{NP}}

\newcommand{\lrun}{\leftarrow}

\newcommand{\la}{\leftarrow}
\newcommand{\ra}{\rightarrow}

\newcommand{\seteq}{\coloneqq}

\newcommand{\cA}{\mathcal{A}}
\newcommand{\cB}{\mathcal{B}}

\newcommand{\cD}{\mathcal{D}}

\newcommand{\cM}{\mathcal{M}}

\newcommand{\cP}{\mathcal{P}}
\newcommand{\cQ}{\mathcal{Q}}

\newcommand{\cS}{\mathcal{S}}

\newcommand{\cV}{\mathcal{V}}

\def\makeuppercase#1{
\expandafter\newcommand\csname tl#1\endcsname{\widetilde{#1}}
}

\def\makelowercase#1{
\expandafter\newcommand\csname tl#1\endcsname{\widetilde{#1}}
}

\newcommand{\N}{\mathbb{N}}

\newcommand{\R}{\mathbb{R}}

\newcommand{\Ms}{\mathcal{M}}

\newcommand{\Cs}{\mathcal{C}}
\newcommand{\Ks}{\mathcal{K}}

\newcommand{\secp}{\lambda}

\newcommand{\adva}[2]{\mathsf{Adv}_{#1}^{\mathsf{#2}}}
\newcommand{\advb}[3]{\mathsf{Adv}_{#1}^{\mathsf{#2} \mbox{-} \mathsf{#3}}}
\newcommand{\advc}[4]{\mathsf{Adv}_{#1}^{\mathsf{#2} \mbox{-} \mathsf{#3} \mbox{-} \mathsf{#4}}}

\newcommand{\expa}[3]{\mathsf{Exp}_{#1}^{ \mathsf{#2} \mbox{-} \mathsf{#3}}}
\newcommand{\expb}[4]{\mathsf{Exp}_{#1}^{ \mathsf{#2} \mbox{-} \mathsf{#3} \mbox{-} \mathsf{#4}}}

\newcommand{\sfhyb}[2]{\mathsf{Hyb}_{#1}^{#2}}

\newcommand*{\sk}{\keys{sk}}

\newcommand{\ck}{\keys{ck}}

\newcommand{\ct}{\keys{CT}}

\newcommand{\CT}{\keys{CT}}

\newcommand*{\msg}{\keys{msg}}

\newcommand*{\keys}[1]{\mathsf{#1}}

\newcommand*{\algo}[1]{\ensuremath{\mathsf{#1}}}

\newcommand{\compclass}[1]{\textbf{\textrm{#1}}}

\newenvironment{boxfig}[2]{\begin{figure}[#1]\fbox{\begin{minipage}{0.97\linewidth}
                        \vspace{0.2em}
                        \makebox[0.025\linewidth]{}
                        \begin{minipage}{0.95\linewidth}
            {{
                        #2 }}
                        \end{minipage}
                        \vspace{0.2em}
                        \end{minipage}}}{\end{figure}}

\newcommand{\bit}{\{0,1\}}

\newcommand{\keygen}{\algo{KeyGen}}

\newcommand{\Enc}{\algo{Enc}}
\newcommand{\Dec}{\algo{Dec}}

\newcommand\SKE{\algo{SKE}}

\newcommand{\ske}{\mathsf{ske}}

\newcommand{\negl}{{\mathsf{negl}}}

\newcommand{\poly}{{\mathrm{poly}}}



%% file: Sec_Introduction.tex

\section{Introduction}
\subsection{Background}\label{sec:background}

Zero-knowledge~\cite{SICOMP:GolMicRac89},
which roughly states that the verifier cannot learn anything beyond the validity of the statement,
is one of the most important concepts in cryptography and computer science. 
The study of zero-knowledge has a long history in classical cryptography, and recently there have been many results
in quantum cryptography.
In known constructions of classical zero-knowledge protocols for $\NP$, either of zero-knowledge or soundness holds only 
against computationally bounded adversaries.
Indeed, achieving both statistical zero-knowledge and statistical soundness at the same time with classical verifier 
is impossible for $\NP$ unless the polynomial-time hierarchy collapses~\cite{STOC:Fortnow87}.
It is also believed to be impossible even with a quantum verifier~\cite{MW18}.

Broadbent and Islam~\cite{TCC:BroIsl20} recently suggested an idea of the novel compromise:
realizing ``everlasting zero-knowledge'' by using quantum encryption with certified deletion.
The everlasting security defined by Unruh~\cite{C:Unruh13} states that the protocol remains secure as long as 
the adversary runs in polynomial-time during the execution of the protocol.
Quantum encryption with certified deletion introduced by Broadbent and Islam~\cite{TCC:BroIsl20} is a new quantum cryptographic primitive where
a classical message is encrypted into a quantum ciphertext,
and the receiver in possession of a quantum ciphertext can generate a classical certificate that shows that
the receiver has deleted the quantum ciphertext.
If the certificate is valid, the receiver can no longer decrypt the message even if it receives the secret key.
Broadbent and Islam’s idea is to use quantum commitment with a similar certified deletion security
to encrypt the first message from the prover to the verifier in the standard $\Sigma$-protocol.
Once the verifier issues the deletion certificate for all commitments that are not opened by the verifier's challenge,
even an unbounded verifier can no longer access the committed values of the unopened commitments.
They left the formal definition and the construction as future works.

There are many obstacles to realizing their idea.
First, their quantum encryption with certified deletion cannot be directly used in a $\Sigma$-protocol because
it does not have any binding property. Their ciphertext consists of a classical and quantum part. The classical part is $m\oplus u\oplus H(r)$, 
where $m$ is the plaintext, $u$ and $r$ are random bit strings, and $H$ is a hash function. The quantum part is
a random BB84 states whose computational basis states encode $r$.
The decryption key is $u$ and the place of computational basis states that encode $r$, 
and therefore it is not binding: by changing $u$, a different message can be obtained.
We therefore need to extend quantum encryption with certified deletion
in such a way that the statistical binding property is included. 

Second, defining a meaningful notion of ``everlasting zero-knowledge proof'' itself is non-trivial.
In fact, everlasting zero-knowledge proofs for $\QMA$ or even for $\NP$ in the sense of Unruh's definition~\cite{C:Unruh13} 
are unlikely to exist.\footnote{We mention that everlasting zero-knowledge \emph{arguments}, which only satisfy computational soundness, can exist. Indeed, any statistical zero-knowledge argument is everlasting zero-knowledge argument. One may think that the computational soundness is fine since that ensures everlasting soundness in the sense of Unruh's definition~\cite{C:Unruh13}. For practical purposes, this may be true. On the other hand, we believe that it is theoretically interesting to pursue (a kind of) everlasting zero-knowledge without compromising the soundness as is done in this paper.}
To see this, recall that the definition of quantum statistical zero-knowledge \cite{SIAM:Wat09,MW18} requires a simulator to simulate the view of a \emph{quantum polynomial-time} malicious verifier in a statistically indistinguishable manner. Therefore, everlasting zero-knowledge in the sense of Unruh's definition~\cite{C:Unruh13} is actually equivalent to quantum statistical zero-knowledge. On the other hand, as already mentioned, it is believed that quantum statistical zero-knowledge proofs for $\NP$ do not exist \cite{MW18}. In particular, Menda and Watrous~\cite{MW18} constructed an oracle relative to which quantum statistical zero-knowledge proofs for (even a subclass of) $\NP$ do not exist.


However, we notice that this argument does not go through for \emph{certified everlasting zero-knowledge}, where
the verifier can issue a classical certificate that shows that the verifier has deleted its information.
Once a valid certificate has been issued, even unbounded malicious verifier can no longer learn anything
beyond the validity of the statement.
The reason is that certified everlasting zero-knowledge does not imply statistical zero-knowledge since it does not ensure any security against a malicious verifier that refuses to provide a valid certificate of deletion.  
Therefore, we have the following question.
\begin{center}
\emph{
Is it possible to define and construct a certified everlasting zero-knowledge proof for $\QMA$?
}
\end{center}

\subsection{Our Results}
In this work, we define and construct the certified everlasting zero-knowledge proof for $\QMA$.
This goal is achieved in the following four steps.
\begin{enumerate}
\item
We define a new quantum cryptographic primitive, which we call
{\it commitment with statistical binding and certified everlasting hiding}. 
In this new commitment scheme, binding is statistical but hiding is computational.
However, the hiding property becomes statistical once the receiver has issued a valid certificate that
shows that the receiver has deleted the committed information.
\item
We construct commitment with statistical binding and certified everlasting hiding.
We use secret-key quantum encryption with certified deletion as the building block in a black box way .
This construction is secure in the quantum random oracle model~\cite{AC:BDFLSZ11}.
\item
We define a new notion of zero-knowledge proof, which we call \emph{the certified everlasting zero-knowledge proof for} $\QMA$.
It is a computational zero-knowledge proof for $\QMA$ with the following additional property. A verifier can issue a classical
certificate that shows that the verifier has deleted its information. If the certificate is valid,
even unbounded malicious verifier can no longer learn anything beyond the validity of the statement.
\item
We apply commitment with statistical binding and certified everlasting hiding 
to the quantum $\Sigma$-protocol for $\QMA$ by Broadbent and Grilo~\cite{FOCS:BroGri20}
to construct the certified everlasting zero-knowledge proof for $\QMA$.
\end{enumerate}

We have three remarks on our results.
First, although our main results are the definition and the construction of the certified everlasting zero-knowledge
proof for $\QMA$, our commitment with statistical binding and certified everlasting hiding itself is also of independent interest.
There will be many other useful applications beyond zero-knowledge.
In fact, it is known that binding and hiding cannot be made statistical at the same time
even in the quantum world~\cite{LC97,Mayers97},
and therefore our new commitment scheme provides a nice compromise.

Second, our new commitment scheme and the new zero-knowledge proof
are the first cryptographic applications of symmetric-key quantum encryption with certified deletion.
Although certified deletion is conceptually very interesting, there was no concrete construction of cryptographic 
applications when it was first introduced~\cite{TCC:BroIsl20}.
One reason why the applications are limited is that 
in cryptography it is not natural to consider the case when the receiver receives the private key later.
Hiroka et al.~\cite{EPRINT:HMNY21} recently extended the symmetric-key scheme by Broadbent and Islam~\cite{TCC:BroIsl20} to a public-key encryption scheme, an attribute-based encryption scheme, and a publicly verifiable scheme, which have opened many applications.
However, one disadvantage is that their security is the computational one unlike the symmetric-key scheme~\cite{TCC:BroIsl20}. 
Therefore it was open whether there is any cryptographic application of the information-theoretically secure certified deletion scheme.
Our results provide the first cryptographic applications of it.
Interestingly, the setup of the symmetric-key scheme~\cite{TCC:BroIsl20}, where the receiver does not have the private key in advance,
nicely fits into the framework of the $\Sigma$-protocol, because the verifier (receiver) in the $\Sigma$-protocol does not have 
the decryption key of the first encrypted message from the prover (sender).

Finally, note that certified everlasting zero-knowledge and certified everlasting hiding
seem to be impossible in the classical world, because a malicious adversary can copy its information.
In particular, certified everlasting zero-knowledge against classical verifiers clearly implies honest-verifier statistical zero-knowledge since an honest verifier runs in polynomial-time.\footnote{A similar argument does not work for quantum verifiers since the honest-verifier quantum statistical zero-knowledge~\cite{FOCS:Watrous02} requires a simulator 
to simulate honest verifier's internal state \emph{at any point} of the protocol execution. This is not implied by certified everlasting zero-knowledge, which only requires security after generating a valid deletion certificate.}    Moreover, it is known that $\compclass{HVSZK}=\compclass{SZK}$ where $\compclass{HVSZK}$ and $\compclass{SZK}$ are languages that have honest-verifier statistical zero-knowledge proofs and (general) statistical zero-knowledge proofs, respectively~\cite{STOC:GolSahVad98}.
Therefore, if certified everlasting zero-knowledge proofs for $\NP$ with classical verification exist, we obtain
$\NP\subseteq\compclass{HVSZK}=\compclass{SZK}$, which means the collapse of the polynomial-time hierarchy~\cite{STOC:Fortnow87}.
Though the above argument only works for protocols in the standard model, no construction of honest-verifier statistical zero-knowledge proofs for $\NP$ is known in the random oracle model either. 
Our results therefore add novel items to the list of quantum cryptographic primitives that can be achieved only in the quantum world.



\subsection{Technical Overview}
\paragraph{Certified everlasting zero-knowledge.}
As explained in \cref{sec:background}, everlasting zero-knowledge proofs for $\NP$ (and for $\QMA$) seem impossible even with quantum verifiers. Therefore, we introduce a relaxed notion of zero-knowledge which we call \emph{certified everlasting zero-knowledge} inspired by quantum encryption with certified deletion introduced by Broadbent and Islam~\cite{TCC:BroIsl20}.
Certified everlasting zero-knowledge ensures security against malicious verifiers that run in polynomial-time during the protocol and provide a valid certificate that  sensitive information is ``deleted''. 
(For the formal definition, see \cref{sec:def_everlasting_ZK}.)
The difference from everlasting zero-knowledge is that it does not ensure security against malicious verifiers that do not provide a valid certificate. We believe that this is still a meaningful security notion since if the verifier refuses to provide a valid certificate, the prover may penalize the verifier for cheating.

\paragraph{Quantum commitment with certified everlasting hiding.}
Our construction of certified everlasting zero-knowledge proofs is based on the idea sketched by Broadbent and 
Islam~\cite{TCC:BroIsl20}. 
(For the details of the construction, see \cref{sec:3_construction}.)
The idea is to implement a $\Sigma$-protocol using a commitment scheme with certified deletion. However, they did not give a construction or definition of commitment with certified deletion. First, we remark that the encryption with certified deletion in \cite{TCC:BroIsl20} cannot be directly used as a commitment. A natural way to use their scheme as a commitment scheme is to consider a ciphertext as a commitment. However, since different secret keys decrypt the same ciphertext into different messages, this does not satisfy the binding property as commitment.

A natural (failed) attempt to fix this problem is to add a classical commitment to the secret key of the encryption scheme with certified deletion making use of the fact that the secret key of the encryption with certified deletion in \cite{TCC:BroIsl20} is classical.
That is, a commitment to a message $m$ consists of 
$$
(\ct=\Enc(\sk,m),\mathsf{com}=\Commit(\sk))
$$
where $\Enc$ is the encryption algorithm of the scheme in \cite{TCC:BroIsl20}, $\sk$ is its secret key, and $\Commit$ is a statistically binding and computationally hiding classical commitment scheme. 
This resolves the issue of binding since the secret key is now committed by the classical commitment scheme. On the other hand, we cannot prove a hiding property that is sufficiently strong for achieving certified everlasting zero-knowledge. It is not difficult to see that what we need here is \emph{certified everlasting hiding}, which ensures that once a receiver generates a valid certificate that it deleted the commitment in a polynomial-time, the hiding property is ensured even if the receiver runs in unbounded-time afterwards. Unfortunately, we observe that the above generic construction seems insufficient for achieving certified everlasting hiding.\footnote{One may think that we can just use statistically hiding commitment. However, such a commitment can only satisfy computational binding, which is not sufficient for achieving certified everlasting zero-knowledge \emph{proofs} rather than arguments.}
The reason is as follows: We want to reduce the certified everlasting hiding to the certified deletion security of $\Enc$. However, the security of $\Enc$ can be invoked only if $\sk$ is information theoretically hidden before the deletion. On the other hand, $\sk$ is committed by a statistically binding commitment in the above construction, and thus $\sk$ is information theoretically determined from the commitment. 
Therefore, we have to somehow delete the information of $\sk$ from the commitment in some hybrid game in a security proof. A similar issue was dealt with by Hiroka et al. \cite{EPRINT:HMNY21} by using receiver non-committing encryption in the context of public key encryption with certified deletion. However, their technique inherently relies on the assumption that an adversary runs in polynomial-time \emph{even after the deletion}. Therefore, their technique is not applicable in the context of certified everlasting hiding.


To overcome the above issue, we rely on random oracles. We modify the above construction as follows:
$$
(\ct=\Enc(\sk,m),\mathsf{com}=\Commit(R),H(R)\oplus \sk)
$$
where $R$ is a sufficiently long random string and $H$ is a hash function modeled as a random oracle whose output length is the same as that of $\sk$. 
We give an intuition on why the above issue is resolved with this modification. As explained in the previous paragraph, we want to delete the information of $\sk$ from the commitment in some hybrid game. By the computational hiding of commitment, a polynomial-time receiver cannot find $R$ from $\Commit(R)$. Therefore, it cannot get any information on $H(R)$ since otherwise we can ``extract'' $R$ from one of receiver's queries.
This argument can be made rigorous by using the one-way to hiding lemma~\cite{JACM:Unruh15,C:AmbHamUnr19}. Importantly, we only have to assume that the receiver runs in polynomial-time \emph{before the deletion} and do not need to assume anything about the running time after the deletion  because we extract $R$ from one of the queries before the deletion. 
Since $\sk$ is masked by $H(R)$, the receiver cannot get any information on $\sk$ either. Thus, we can simulate the whole commitment $(\CT,\cment,H(R)\oplus \sk)$ without using $\sk$, which resolves the issue and enables us to reduce certified everlasting hiding to certified deletion security of $\Enc$.

We remark that quantum commitments in general cannot satisfy the binding property in the classical sense.
Indeed, if a malicious sender generates a superposition of valid commitments on different messages $m_0$ and $m_1$, it can later open to $m_0$ or $m_1$ with probability $1/2$ for each. Defining a binding property for quantum commitments is non-trivial, and there have been proposed various flavors of definitions in the literature, e.g., \cite{TCC:CDMS04,C:DamFehSal04,C:DFRSS07,Yan20a,BB21}. It might be possible to adopt some of those definitions.
However, we choose to introduce a new definition, which we call 
\emph{classical-extractor-based binding}, 
tailored to our construction because this is more convenient for our purpose. 
Classical-extractor-based binding captures the property of our construction that the randomness $R$ is information-theoretically determined by the classical part $\mathsf{com}=\Commit(R)$ of a commitment, and the decommitment can be done by using the rest part of the commitment and $R$.\footnote{For this definition to make sense, we need to require that $\mathsf{com}=\Commit(R)$ is classical. This can be ensured if the honest receiver measures it as soon as receiving it even if only quantum communication channel is available.}  
In particular, this roughly means that one can extract the committed message with an unbounded-time extractor before the sender decommits.
This enables us to prove soundness for our certified everlasting zero-knowledge proofs in essentially the same manner as in the classical case.

The details of the construction and security proofs are explained in \cref{sec:bit_comm_construction}.


\paragraph{Certified everlasting zero-knowledge proof for QMA.}
Once we obtain a commitment scheme with certified everlasting hiding, the construction of certified everlasting zero-knowledge proofs is straightforward based on the idea sketched in \cite{TCC:BroIsl20}. Though they only considered a construction for $\NP$, we observe that the idea can be naturally extended to a construction for $\QMA$ since a ``quantum version'' of $\Sigma$-protocol for $\QMA$ called $\Xi$-protocol is constructed by Broadbent and Grilo~\cite{FOCS:BroGri20}.
Below, we sketch the construction for clarity. 
Let $A=(A_{\yes},A_{\no})$ be a promise problem in $\QMA$.
\cite{FOCS:BroGri20} showed that for any $\statement\in A_{\yes}$ and any corresponding witness $\witness$, 
it is possible to generate (in a quantum polynomial-time)
so-called the local simulatable history state $\rho_\hist$ from $\witness$, which 
satisfies the following two special properties (for details, see \cref{def:k-SimQMA}):
\begin{enumerate}
\item[(LS1)]
The verification can be done by measuring randomly chosen five qubits of $\rho_\hist$.
\item[(LS2)]
The classical description of any five-qubit reduced density matrix of $\rho_\hist$ can be obtained
in classical polynomial-time.
\end{enumerate}
With these properties, the quantum $\Sigma$-protocol of \cite{FOCS:BroGri20} is constructed as follows:
\begin{enumerate}
\item[{\bf 1. Commitment phase}:]
The prover randomly chooses $x,z\in\bit^n$,
and sends $(X^xZ^z \rho_\hist Z^zX^x)\otimes \cment(x,z)$ to the verifier, 
where 
$X^xZ^z\seteq \prod_{i=1}^n X_i^{x_i}Z_i^{z_i}$,  
$n$ is the number of qubits of $\rho_\hist$,
and $\cment(x,z)$ is a classical commitment of $(x,z)$. 
\item[{\bf 2. Challenge phase}:]
The verifier randomly chooses a subset $S\subset[n]$ of size $|S|=5$, and sends it to the prover.
\item[{\bf 3. Response phase}:]
The prover opens the commitment for $\{x_i,z_i\}_{i\in S}$.
\item[{\bf 4. Verification phase}:]
The verifier applies $\prod_{i\in S}X_i^{x_i}Z_i^{z_i}$ on the state and measures qubits in $S$.
\end{enumerate}
The correctness and the soundness come from the property (LS1),
and the zero-knowledge comes from the property (LS2).
If the classical commitment scheme used in the above construction is the one with statistical binding and
computational hiding, the quantum $\Sigma$-protocol is a computational zero-knowledge proof for $\QMA$,
because the unbounded malicious verifier can open the commitment of $\{x_i,z_i\}_{i \in [n]\setminus S}$, and therefore can obtain the entire $\rho_\hist$.
If more than five qubits of $\rho_\hist$ are available to the malicious verifier, the zero-knowledge property no longer holds.

We construct the certified everlasting zero-knowledge proof for $\QMA$ based on the quantum $\Sigma$-protocol.
Our idea is to use commitment with certified everlasting hiding and statistical binding
for the commitment of $(x,z)$ in the above construction of the quantum $\Sigma$-protocol.
If the verifier issues a valid deletion certificate for the commitment of $\{x_i,z_i\}_{i\in [n]\setminus S}$,
even unbounded malicious verifier can no longer learn $\{x_i,z_i\}_{i\in [n]\setminus S}$, 
and therefore what it can access is only the five qubits of $\rho_\hist$.
This gives a proof for certified everlasting zero-knowledge. 
Using classical-extractor-based binding, the proof of statistical soundness can be done almost in the same way as in \cite{FOCS:BroGri20}.   
Recall that classical-extractor-based binding enables us to extract the committed message with an unbounded-time extractor before the sender decommits. Therefore, we can extract the committed $(x,z)$ from $\mathsf{com}(x,z)$. Since the extraction is done before the challenge phase, the extracted values do not depend on the challenge $S$. Then, it is easy to reduce the soundness of the scheme to that of the the original $\QMA$ promise problem $A$.    
The details of the construction is explained in \cref{sec:3_construction}.

\subsection{Related Works}

\paragraph{Zero-knowledge for QMA.}
Zero-knowledge for $\QMA$ was first constructed by
Broadbent, Ji, Song, and Watrous \cite{FOCS:BJSW16}.
Broadbent and Grilo \cite{FOCS:BroGri20} gave an elegant and simpler construction what they call the
$\Xi$-protocol (which is considered as a quantum version of the standard $\Sigma$-protocol) by using the local simulatability~\cite{FOCS:GriSloYue19}. Our construction is based on the $\Xi$-protocol.
Bitansky and Shmueli \cite{STOC:BitShm20} gave the first constant round zero-knowledge argument for $\QMA$ with negligible soundness error. 
Brakerski and Yuen \cite{BraYue} gave a construction of $3$-round \emph{delayed-input} zero-knowledge proof for $\QMA$ where the prover needs to know the statement and witness only for generating its last message. 
Chardouvelis and Malavolta~\cite{EPRINT:ChaMal21} constructed 4-round statistical zero-knowledge arguments for $\QMA$ and 2-round zero-knowledge for $\QMA$ in the timing model.

Regarding non-interactive zero-knowledge proofs or arguments (NIZK),
Kobayashi \cite{Kobayashi03} first studied (statistically sound and zero-knowledge) NIZKs in a model where the prover and verifier share Bell pairs, and gave a complete problem in this setting.
It is unlikely that the complete problem contains (even a subclass of) $\NP$ \cite{MW18}, and thus even a NIZK for all $\NP$ languages is unlikely to exist in this model.
Chailloux et al. \cite{TCC:CCKV08} showed that there exists a (statistically sound and zero-knowledge) NIZK for all languages in $\compclass{QSZK}$ in the help model where a trusted party generates a pure state \emph{depending on the statement to be proven} and gives copies of the state to both prover and verifier.
Recently, there are many constructions of NIZK proofs or arguments for $\QMA$ in various kind of setup models and assumptions \cite{TCC:ACGH20,C:ColVidZha20,FOCS:BroGri20,C:Shmueli21,C:BCKM21a,EPRINT:MorYam21,BartusekMalavolta}.   

\paragraph{Quantum commitment.}
It is well-known that statistically binding and hiding commitments are impossible even with quantum communication \cite{LC97,Mayers97}. 
On the other hand, there are a large body of literature on constructing quantum commitments assuming some computational assumptions, e.g., see the references in the introduction of \cite{Yan20a}. Among them, several works showed the possibility of using quantum commitments in constructions of zero-knowledge proofs and arguments \cite{ISAAC:YWLQ15,FUWYZ20,Yan20a,BB21}. However, they only consider replacing classical commitments with quantum commitments in classical constructions while keeping the same functionality and security level as the classical construction. In particular, none of them considers protocols for $\QMA$ or properties that are classically impossible to achieve like our notion of the certified everlasting zero-knowledge.

%% file: Sec_Preliminaries.tex
\section{Preliminaries}\label{sec:preliminaries}

\subsection{Notations}\label{sec:notation}
 
Here we introduce basic notations we will use.
In this paper, $x\leftarrow X$ denotes selecting an element  from a finite set $X$ uniformly at random, 
and $y\leftarrow A(x)$ denotes assigning to $y$ the output of a probabilistic or deterministic algorithm $A$ on an input $x$.
When we explicitly show that $A$ uses randomness $r$, we write $y\leftarrow A(x;r)$.
When $D$ is a distribution, $x\leftarrow D$ denotes sampling an element from $D$.
Let $[n]$ be the set $\{1,\dots,n\}$. Let $\lambda$ be a security parameter, and $y\seteq z$ denotes that $y$ is set, defined, or substituted by $z$.
For a bit string $s\in\{0,1\}^n$, $s_i$ denotes the $i$-th bit of $s$. QPT stands for quantum polynomial time.
PPT stands for (classical) probabilistic polynomial time.
For a subset $S\subseteq W$ of a set $W$, $\overline{S}$ is the complement of $S$, i.e., 
$\overline{S}\seteq W\setminus S$. 
A function $f: \N \ra \R$ is a negligible function if for any constant $c$, there exists $\secp_0 \in \N$ such that for any $\secp>\secp_0$, $f(\secp) < \secp^{-c}$. We write $f(\secp) \leq \negl(\secp)$ to denote $f(\secp)$ being a negligible function.

\subsection{Quantum Computation}\label{sec:quantum_computation}
We assume the familiarity with basics of quantum computation, and use standard notations. 
Let us denote $\cQ$ be the state space of a single qubit.
$I$ is the two-dimensional identity operator. For simplicity, we often write
$I^{\otimes n}$ as $I$ for any $n$ when the dimension of the identity operator is clear from the context.
For any single-qubit operator $O$, $O_i$ means an operator that applies $O$ on the $i$-th qubit and applies $I$ on all other qubits.
$X$ and $Z$ are the Pauli $X$ and $Z$ operators, respectively.
For any $n$-bit strings $x\seteq(x_1,x_2,\cdots,x_n)\in\{0,1\}^n$ and $z\seteq (z_1,z_2,\cdots,z_n)\in\{0,1\}^n$,
$X^x\seteq\prod_{i\in[n]}X_i^{x_i}$ and 
$Z^z\seteq\prod_{i\in[n]}Z_i^{z_i}$.
For any subset $S$, $\Tr_S$ means the trace over all qubits in $S$.
For any quantum state $\rho$ and a bit string $s\in\bit^n$, $\rho\otimes s$ means $\rho\otimes|s\rangle\langle s|$.
The trace distance between two states $\rho$ and $\sigma$ is given by 
$\frac{1}{2}\norm{\rho-\sigma}_{\tr}$, where $\norm{A}_{\tr}\seteq \Tr \sqrt{{\it A}^{\dagger}{\it A}}$ is the trace norm. 
If $\frac{1}{2}\norm{\rho-\sigma}_{\tr}\leq \epsilon$, we say that $\rho$ and $\sigma$ are $\epsilon$-close.
If $\epsilon=\negl(\lambda)$, then we say that $\rho$ and $\sigma$ are statistically indistinguishable.

Let $C_0$ and $C_1$ be quantum channels from $p$ qubits to $q$ qubits, where $p$ and $q$ are polynomials. 
We say that they are computationally indistinguishable, and denote it by $C_0\approx_c C_1$
if there exists a negligible function $\negl$ such that  
$|\Pr[D((C_0\otimes I)(\sigma))=1]-\Pr[D((C_1\otimes I)(\sigma))=1] |\leq \negl(\lambda)$
for any polynomial $k$, any $(p+k)$-qubit state $\sigma$, and any polynomial-size quantum circuit $D$ acting on 
$q+k$ qubits.
We say that $C_0$ and $C_1$ are statistically indistinguishable, and denote it by $C_0\approx_s C_1$,
if $D$ is an unbounded algorithm.

\begin{lemma}[Quantum Rewinding Lemma~\cite{SIAM:Wat09}]\label{lemma:rewinding}
Let $Q$ be a quantum circuit that acts on an $n$-qubit state $|\psi\rangle$ and an $m$-qubit auxiliary state $|0^m\rangle$.
Let
    $p(\psi)\seteq||(\langle 0|\otimes I)Q(|\psi\rangle \otimes |0^m\rangle)||^2$ and
    $
    |\phi(\psi)\rangle\seteq\frac{1}{\sqrt{p(\psi)}} (\langle 0|\otimes I) Q(|\psi\rangle\otimes |0^m\rangle).
    $
Let $p_0,q\in(0,1)$ and $\epsilon\in (0,\frac{1}{2})$ such that $|p(\psi)-q|<\epsilon$, $p_0(1-p_0)<q(1-q)$, and $p_0<p(\psi)$.
Then there is a quantum circuit $R$ of size at most 
$
    O\left(\frac{\log(\frac{1}{\epsilon}){\rm size}(Q)}{p_0(1-p_0)}\right)
$
such that on input $|\psi\rangle$, $R$ computes a quantum state $\rho(\psi)$ that satisfies 
$
    \langle\phi(\psi)|\rho(\psi)|\phi(\psi)\rangle\geq 1-16\epsilon\frac{\log^2(\frac{1}{\epsilon})}{p_0^2(1-p_0)^2}.
$
\end{lemma}

\begin{lemma}[One-Way to Hiding Lemma \cite{C:AmbHamUnr19}]\label{lemma:one-way_to_hiding}
Let $S\subseteq \mathcal{X}$ be a random subset of $\mathcal{X}$. Let $G,H:\mathcal{X}\rightarrow\mathcal{Y}$ be random functions satisfying $\forall x\notin S$ $[G(x)=H(x)]$. Let $z$ be a random classical bit string. 
($S,G,H,z$ may have an arbitrary joint distribution.)
Let $\cA$ be an oracle-aided quantum algorithm that makes at most $q$ quantum queries.
Let $\cB$ be an algorithm that on input $z$ chooses $i\leftarrow[q]$, runs $\cA^{H}(z)$, measures $\cA$'s $i$-th query, and outputs the measurement outcome.
Then we have
$
    \abs{\Pr[\cA^G(z)=1]-\Pr[\cA^H(z)=1]}\leq2q\sqrt{\Pr[\cB^H(z)\in S]}.
$
\end{lemma}

\subsection{QMA and $k$-$\compclass{SimQMA}$}

\begin{definition}[\compclass{QMA}]\label{def:QMA}
We say that a promise problem $A=(A_{\yes},A_{\no})$ is in $\QMA$ if there exist a polynomial $p$, a QPT algorithm $V$,
and $0\le\beta<\alpha\le1$ with $\alpha-\beta\geq\frac{1}{\poly(|\statement|)}$ 
such that
\begin{description}
    \item[Completeness:]
    For any $\statement\in A_{\yes}$, there exists a quantum state $\witness$ of $p(|\statement|)$-qubit (called a witness) such that
    \begin{align}
    \Pr[V(\statement,\witness)=\top]\geq \alpha.
    \end{align}
    
    \item[Soundness:]
    For any $\statement\in A_{\no}$ and any quantum state $\witness$ of $p(|\statement|)$-qubit, 
    \begin{align}
    \Pr[V(\statement,\witness)=\top]\leq \beta.
    \end{align}
\end{description}
For any $\statement\in A_{\yes}$, $R_A(\statement)$ is the (possibly infinite) set of all quantum states 
$\witness$ such that $\Pr[V(\statement,\witness)=\top]\geq \frac{2}{3}$.
\end{definition}


A complexity class of $k$-$\compclass{SimQMA}$ is introduced, and proven to be equal to $\QMA$ in \cite{FOCS:BroGri20}.
\begin{definition}[$k$-\compclass{SimQMA}~\cite{FOCS:BroGri20}]\label{def:k-SimQMA}
A promise problem $A=(A_{\yes},A_{\no})$ is in $k$-$\compclass{SimQMA}$ with soundness $\beta(|\statement|)\leq 1-\frac{1}{\poly(|\statement|)}$, 
if there exist polynomials $m$ and $n$ such that given $\statement\in A_{\yes}$,
there is an efficient deterministic algorithm that computes $m(|\statement|)$ $k$-qubit POVMs $\{\Pi_1,I-\Pi_1\},\dots,\{\Pi_{m(|\statement|)},I-\Pi_{m(|\statement|)}\}$ such that:
\begin{description}
    \item [Simulatable completeness:]
    If $\statement\in A_{\yes}$, there exists an $n(|\statement|)$-qubit state $\rho_{\hist}$, which we call a simulatable witness, such that for all $c\in[m]$,
        $\mathrm{Tr}(\Pi_c\rho_{\hist})\geq 1-\negl(|\statement|)$,
    and 
    there exists a set of $k$-qubit density matrices $\{\rho_\simulator^{\statement,S}\}_{S\subseteq[n(|\statement|)],|S|=k}$ 
    that can be computed in polynomial time from $\statement$ and $\rho_{\hist}$
    such that 
    $
        ||\Tr_{\overline{S}}(\rho_{\hist})-\rho_\simulator^{\statement,S}||_{\rm tr}\leq \negl(|\statement|).
        $
    \item [Soundness:]
    If $\statement\in A_{\no}$, for any $n(|\statement|)$-qubit state $\rho$, 
    $
        \frac{1}{m}\sum_{c\in[m]} \Tr(\Pi_{c}\rho)\leq \beta(|\statement|).
        $
\end{description}
\end{definition}

\subsection{Cryptographic Tools}\label{sec:crypt_tool}
In this section, we review cryptographic tools used in this paper.

\paragraph{Non-interactive commitment.}
\begin{definition}[Non-Interactive Commitment (Syntax)]\label{def:Non_interactive_commitments}
Let $\lambda$ be the security parameter and let $p$, $q$ and $r$ be some polynomials.
A (classical) non-interactive commitment scheme consists of a single PPT algorithm $\Commit$ with plaintext space $\Ms\seteq\{0,1\}^{p(\lambda)}$, randomness space $\{0,1\}^{q(\lambda)}$ and commitment space $\Cs\seteq \{0,1\}^{r(\lambda)}$ satisfying two properties: 
\begin{description}
\item[Perfect binding:]For every $(r_0,r_1)\in\{0,1\}^{q(\lambda)}\times\bit^{q(\lambda)}$ and 
$(m,m')\in \Ms^2$ such that $m\neq m'$,
we have that $\Commit(m;r_0)\neq \Commit(m';r_1)$, where $(\Commit(m;r_0), \Commit(m';r_1))\in \Cs^2$.

\item[Unpredictability:] Let $\Sigma\seteq\Commit$. For any QPT adversary $\cA$,  we define the following security experiment $\mathsf{Exp}_{\Sigma,\cA}^{\mathsf{unpre}}(\secp)$.
\begin{enumerate}
\item The challenger chooses $R\lrun \Ms$ and $R'\leftarrow\{0,1\}^{q(\lambda)}$, computes $\cment\leftarrow\Commit(R;R')$, and sends $\cment$ to $\cA$.
\item $\cA$ outputs $R^*$. The output of the experiment is $1$ if $R^*=R$. Otherwise, the output of the experiment is $0$.
\end{enumerate}
We say that the commitment is unpredictable if for any QPT adversary $\cA$, it holds that
\begin{align}
    \adva{\Sigma,\cA}{unpre}(\secp)\coloneqq \abs{\Pr[ \mathsf{Exp}_{\Sigma,\cA}^{\mathsf{unpre}}(\secp)=1] }\leq \negl(\secp).
\end{align}
\end{description}
\end{definition}

\begin{remark}
The unpredictability is a weaker version of computational hiding. We define unpredictability instead of computational hiding since this suffices for our purpose. 
\end{remark}
A non-interactive commitment scheme that satisfies the above definition exists assuming the existence of injective one-way functions or perfectly correct public key encryption \cite{EPRINT:LomSch19}. Alternatively, we can also instantiate it based on random oracles.  


\paragraph{Quantum encryption with certified deletion.} 
Broadbent and Islam~\cite{TCC:BroIsl20} introduced the notion of quantum encryption with certified deletion. 

\begin{definition}[One-Time SKE with Certified Deletion (Syntax)]\label{def:sk_cert_del}
Let $\lambda$ be the security parameter and let $p$, $q$ and $r$ be some polynomials.
A one-time secret key encryption scheme with certified deletion consists of a tuple of algorithms $(\keygen,\Enc,\Dec,\Delete,\Verify)$ with 
plaintext space $\Ms:=\{0,1\}^n$, ciphertext space $\Cs:= \cQ^{\otimes p(\lambda)}$, key space $\Ks:=\{0,1\}^{q(\lambda)}$ and deletion certificate space $\mathcal{D}:= \{0,1\}^{r(\lambda)}$.
\begin{description}
    \item[$\keygen (1^\secp) \ra \sk$:] The key generation algorithm takes as input the security parameter $1^\secp$, and outputs a secret key $\sk \in \Ks$.
    \item[$\Enc(\sk,m) \ra \ct$:] The encryption algorithm takes as input $\sk$ and a plaintext $m\in\Ms$, and outputs a ciphertext $\ct\in \Cs$.
    \item[$\Dec(\sk,\ct) \ra m^\prime~or~\bot$:] The decryption algorithm takes as input $\sk$ and $\ct$, and outputs a plaintext $m^\prime \in \Ms$ or $\bot$.
    \item[$\Delete(\ct) \ra \cert$:] The deletion algorithm takes as input $\ct$, and outputs a certification $\cert\in\mathcal{D}$.
    \item[$\Verify(\sk,\cert)\ra \top~or~\bot$:] The verification algorithm takes $\sk$ and $\cert$, and outputs $\top$ or $\bot$.
\end{description}
\end{definition}

\begin{definition}[Correctness for One-Time SKE with Certified Deletion]\label{def:sk_cd_correctness}
There are two types of correctness. One is decryption correctness and the other is verification correctness.
\begin{description}
\item[Decryption correctness:] There exists a negligible function $\negl$ such that for any $\secp\in \N$ and $m\in\Ms$, 
\begin{align}
\Pr\left[
\Dec(\sk,\ct)= m
\ \middle |
\begin{array}{ll}
\sk\lrun \keygen(1^\secp)\\
\ct \lrun \Enc(\sk,m)
\end{array}
\right] 
\geq1-\negl(\secp).
\end{align}

\item[Verification correctness:] There exists a negligible function $\negl$ such that for any $\secp\in \N$ and $m\in\Ms$, 
\begin{align}
\Pr\left[
\Verify(\sk,\cert)=\top
\ \middle |
\begin{array}{ll}
\sk\lrun \keygen(1^\secp)\\
\ct \lrun \Enc(\sk,m)\\
\cert \lrun \Delete(\ct)
\end{array}
\right] 
\geq
1-\negl(\secp).
\end{align}
\end{description}
\end{definition}

\begin{definition}[Certified Deletion Security for One-Time SKE]\label{def:sk_certified_del}
Let $\Sigma=(\keygen, \Enc, \Dec, \Delete, \Verify)$ be a secret key encryption with certified deletion.
We consider the following security experiment $\expb{\Sigma,\cA}{otsk}{cert}{del}(\secp,b)$.

\begin{enumerate}
    \item The challenger computes $\sk \la \keygen(1^\secp)$.
    \item $\cA$ sends $(m_0,m_1)\in\cM^2$ to the challenger.
    \item The challenger computes $\ct_b \la \Enc(\sk,m_b)$ and sends $\ct_b$ to $\cA$.
    \item $\cA$ sends $\cert$ to the challenger.
    \item The challenger computes $\Verify(\sk,\cert)$. If the output is $\bot$, the challenger sends $\bot$ to $\cA$.
    If the output is $\top$, the challenger sends $\sk$ to $\cA$. 
    \item $\cA$ outputs $b'\in \bit$.
\end{enumerate}
We say that the $\Sigma$ is OT-CD secure if for any unbounded $\cA$, it holds that
\begin{align}
\advc{\Sigma,\cA}{otsk}{cert}{del}(\secp)
\seteq \abs{\Pr[ \expb{\Sigma,\cA}{otsk}{cert}{del}(\secp, 0)=1] - \Pr[ \expb{\Sigma,\cA}{otsk}{cert}{del}(\secp, 1)=1] }\leq \negl(\secp).
\end{align}
\end{definition}

Broadbent and Islam~\cite{TCC:BroIsl20} showed that one-time SKE scheme with certified deletion that satisfies the above correctness and security exists unconditionally.

%% file: Sec_Commitment.tex

\section{Commitment with Certified Everlasting Hiding and Classical-Extractor-Based Binding}\label{sec:everlasting_commitment}

In this section, we define and construct commitment with certified everlasting hiding and
statistical binding.
We adopt a non-standard syntax for the verification algorithm and a slightly involved definition for the binding, which we call the classical-extractor-based binding, that are tailored to our construction. This is 
because they are convenient for our construction of certified everlasting zero-knowledge proof for $\QMA$ given in \cref{sec:ZK}.
We can also construct one with a more standard syntax of verification and binding property, namely, the sum-binding, by essentially the same construction. The detail is given in \cref{app:sum_binding}.

\subsection{Definition}\label{sec:bitcomm_def}

\begin{definition}[Commitment with Certified Everlasting Hiding and Classical-Extractor-Based Binding (Syntax)]\label{def:commitment_with_cd_effective}
Let $\lambda$ be the security parameter and let $p$, $q$, $r$, $s$ and $t$ be some polynomials.
Commitment with certified everlasting hiding and classical-extractor-based binding consists of a tuple of algorithms $(\Commit,\Verify,\Delete,\Cert)$ 
with message space $\Ms:=\{0,1\}^n$, commitment space $\Cs:=\cQ^{\otimes p(\lambda)}\times \{0,1\}^{q(\lambda)}$,
decommitment space $\mathcal{D}:=\{0,1\}^{r(\lambda)}$, key space $\Ks:=\{0,1\}^{s(\lambda)}$ and
deletion certificate space $\mathcal{E}:= \{0,1\}^{t(\lambda)}$.

\begin{description}
\item[$\Commit (1^{\lambda},m)\rightarrow (\cment,\dment,\ck)$:] 
The commitment algorithm takes as input a security parameter $1^{\lambda}$ and a message $m\in\cM$, and outputs a commitment $\cment\in\Cs $,
a decommitment $\dment\coloneqq(\dment_1,\dment_2)\in\mathcal{D}$ and a key $\ck\in\Ks$.
Note that $\cment$ consists of a quantum state $\psi\in\cQ^{\otimes p(\lambda)}$ and a classical bit string $f\in\{0,1\}^{q(\lambda)}$.
\item[$\Verify(\cment,\dment)\rightarrow m'~or~\bot $:]
The verification algorithm consists of two algorithms, $\Verify_1$ and $\Verify_2$.
It parses $\dment=(\dment_1,\dment_2)$.
$\Verify_1$ takes $\cment$ and $(\dment_1,\dment_2)$ as input, and outputs $\top$ or $\bot$.
$\Verify_2$ takes $\cment$ and $\dment_1$ as input, and outputs $m'$.
If the output of $\Verify_1$ is $\bot$, then the output of $\Verify$ is $\bot$.
Otherwise the output of $\Verify$ is $m'$.
\item[$\Delete(\cment)\rightarrow \cert$:]The deletion algorithm takes $\cment$ as input, and outputs a certificate $\cert\in \mathcal{E}$. 
\item[$\Cert(\cert,\ck)\rightarrow \top~or~\bot$:]The certification algorithm takes $\cert$ and $\ck$ as input, and outputs $\top$ or $\bot$.
\end{description}
\end{definition}

\begin{definition}[Correctness]\label{def:correctness_bit_commitment_effective}
There are two types of correctness, namely, decommitment correctness and deletion correctness.
\begin{description}
\item[Decommitment correctness:]
There exists a negligible function $\negl$ such that for any $\lambda\in\mathbb{N}$ and $m\in\Ms$,
\begin{align}
    \Pr[m\leftarrow\Verify (\cment,\dment)\mid(\cment,\dment,\ck)\leftarrow\Commit(1^{\lambda},m)]\geq1-\negl(\lambda).
\end{align}

\item[Deletion correctness:]
    There exists a negligible function $\negl$ such that for any $\lambda \in\mathbb{N}$ and $m\in\Ms$, 
\begin{align}
    \Pr[\top\leftarrow \Cert(\cert,\ck)\mid(\cment,\dment,\ck)\leftarrow\Commit(1^\lambda,m),\cert\leftarrow\Delete(\cment)]\geq1-\negl(\lambda).
\end{align}
\end{description}
\end{definition}

\begin{definition}[Classical-Extractor-Based Binding]\label{def:effevtive_binding}
There exists an unbounded-time deterministic algorithm $\mathsf{Ext}$ that 
takes $f \in \bit^{q(\lambda)}$ of $\cment$ as input, and outputs
$\dment^{*}_1\leftarrow\mathsf{Ext}(f)$ 
such that
for any $\cment$, any $\dment_1\neq\dment^*_{1}$, and any $\dment_2$, 
$
\Pr[\Verify(\cment,\dment=(\dment_1,\dment_2))=\bot]= 1.
$
\end{definition}

\begin{definition}[Computational Hiding]
Let $\Sigma\seteq(\Commit,\Verify,\Delete,\Cert)$.
Let us consider the following security experiment $\expa{\Sigma,\cA}{c}{hide}(\lambda,b)$ against
any QPT adversary $\mathcal{A}$.
\begin{enumerate}
\item $\cA$ generates $(m_0,m_1)\in \Ms^2$ and sends them to the challenger.
\item The challenger computes $(\cment,\dment,\ck)\leftarrow \Commit(1^{\lambda},m_b)$, and sends $\cment$ to $\mathcal{A}$.
\item $\cA$ outputs $b'\in\{0,1\}$.
\item The output of the experiment is $b'$. 
\end{enumerate}
Computational hiding means that the following is satisfied for any QPT $\cA$.
\begin{align}
\advb{\Sigma,\cA}{c}{hide}(\lambda)\seteq
\left|\Pr[ \expa{\Sigma,\cA}{c}{hide}(\lambda,0)=1]-
\Pr[\expa{\Sigma,\cA}{c}{hide}(\lambda,1)=1]\right|\leq \negl(\lambda).
\end{align}
\end{definition}

\begin{definition}[Certified Everlasting Hiding]
Let $\Sigma\seteq(\Commit,\Verify,\Delete,\Cert)$.
Let us consider the following security experiment $\expa{\Sigma,\cA}{ever}{hide}(\lambda,b)$ against
$\mathcal{A}=(\mathcal{A}_1,\mathcal{A}_2)$ consisting of any QPT adversary $\mathcal{A}_1$ and any unbounded adversary $\mathcal{A}_2$.

\begin{enumerate}
\item $\cA_1$ generates $(m_0,m_1)\in \Ms^2$ and sends it to the challenger.
\item The challenger computes $(\cment,\dment,\ck)\leftarrow \Commit(1^{\lambda},m_b)$, and sends $\cment$ to $\mathcal{A}_1$.
\item At some point, $\mathcal{A}_1$ sends $\cert$ to the challenger, and sends its internal state to $\mathcal{A}_2$.
\item The challenger computes $\Cert(\cert,\ck)$. If the output is $\top$, then the challenger outputs $\top$,
and sends $(\dment,\ck)$ to $\cA_2$.
Else, the challenger outputs $\bot$, and sends $\bot$ to $\cA_2$.
\item $\cA_2$ outputs $b'\in\{0,1\}$.
\item If the challenger outputs $\top$, then the output of the experiment is $b'$. Otherwise, the output of the experiment is $\bot$.
\end{enumerate}
We say that it is certified everlasting hiding if the following is satisfied for any
$\cA=(\cA_1,\cA_2)$.
\begin{align}
\advb{\Sigma,\cA}{ever}{hide}(\lambda)\coloneqq
\left|\Pr[ \expa{\Sigma,\cA}{ever}{hide}(\lambda,0)=1]-
\Pr[\expa{\Sigma,\cA}{ever}{hide}(\lambda,1)=1]\right|\leq \negl(\lambda).
\end{align}
\end{definition}

\begin{remark}
We remark that certified everlasting hiding does not imply computational hiding since it does not require anything if the adversary does not send a valid certificate.
\end{remark}

The following lemma will be used in the construction of the certified everlasting zero-knowledge proof for $\QMA$
in \cref{sec:ZK}.
It is shown with the standard hybrid argument (see \cref{sec:proof_of_bit}).
It is also easy to see that a similar lemma holds for computational hiding.
\begin{lemma}\label{lemma:eachbit}
Let $\Sigma\seteq(\Commit,\Verify,\Delete,\Cert)$ and $\cM=\bit$. 
Let us consider the following security experiment 
$\expb{\Sigma,\cA}{bit}{ever}{hide}(\secp, b)$ 
against
$\mathcal{A}=(\mathcal{A}_1,\mathcal{A}_2)$ consisting of any QPT adversary $\mathcal{A}_1$ and any unbounded adversary $\mathcal{A}_2$.
\begin{enumerate}
\item $\cA_1$ generates $(m^0,m^1)\in \bit^n\times\bit^n$ and sends it to the challenger.
\item The challenger computes 
\begin{align}
(\cment_i(m_i^b),\dment_i(m_i^b),\ck_i(m_i^b))\leftarrow \Commit(1^{\lambda},m^b_i)
\end{align}
for each $i\in[n]$,
and sends $\{\cment_i(m_i^b)\}_{i\in[n]}$ to $\mathcal{A}_1$.
Here, $m_i^b$ is the $i$-th bit of $m^b$.
\item At some point, $\mathcal{A}_1$ sends $\{\cert_i\}_{i\in[n]}$ to the challenger, 
and sends its internal state to $\mathcal{A}_2$.
\item The challenger computes $\Cert(\cert_i,\ck_i(m_i^b))$ for each $i\in[n]$. 
If the output is $\top$ for all $i\in[n]$, then the challenger outputs $\top$,
and sends $\{\dment_i(m_i^b),\ck_i(m_i^b)\}_{i\in[n]}$ to $\cA_2$.
Else, the challenger outputs $\bot$, and sends $\bot$ to $\cA_2$.
\item $\cA_2$ outputs $b'\in\{0,1\}$.
\item If the challenger outputs $\top$, then the output of the experiment is $b'$. Otherwise, the output of the experiment is $\bot$.
\end{enumerate}
If $\Sigma$ is certified everlasting hiding, 
\begin{align}
\advc{\Sigma,\cA}{bit}{ever}{hide}(\lambda)
\coloneqq
\left|\Pr[ \expb{\Sigma,\cA}{bit}{ever}{hide}(\lambda,0)=1]-
\Pr[\expb{\Sigma,\cA}{bit}{ever}{hide}(\lambda,1)=1]\right|\leq \negl(\lambda)
\end{align}
for any $\cA=(\cA_1,\cA_2)$.
\end{lemma}

\subsection{Construction}\label{sec:bit_comm_construction}
Let $\lambda$ be the security parameter, and let $p$, $q$, $r$, $s$, $t$ and $u$ be some polynomials.
We construct commitment with certified everlasting hiding and classical-extractor-based binding,
$\Sigma_{\mathsf{ccd}}=(\Commit,\Verify,\Delete,\Cert)$, with message space $\Ms=\{0,1\}^n$, 
commitment space $\Cs=\cQ^{\otimes p(\lambda)}\times \{0,1\}^{q(\lambda)}\times \{0,1\}^{r(\lambda)}$,
decommitment space $\mathcal{D}=\{0,1\}^{s(\lambda)}\times \{0,1\}^{t(\lambda)}$,
key space $\Ks=\{0,1\}^{r(\lambda)}$ 
and deletion certificate space $\mathcal{E}=\{0,1\}^{u(\lambda)}$
from the following primitives:
\begin{itemize}
\item Secret-key encryption with certified deletion, 
$\Sigma_{\mathsf{skcd}}=\mathsf{SKE}.(\keygen,\Enc,\Dec,\Delete,\Verify)$, with plaintext space $\Ms=\{0,1\}^n$, ciphertext space $\Cs=\cQ^{\otimes p(\lambda)}$, key space $\Ks=\{0,1\}^{r(\lambda)}$, and deletion certificate space $\mathcal{E}=\{0,1\}^{u(\lambda)}$.
\item
Classical non-interactive commitment, $\Sigma_{\cment}=\Classical.\Commit$, with plaintext space $\{0,1\}^{s(\lambda)}$, randomness space $\{0,1\}^{t(\lambda)}$, and commitment space $\{0,1\}^{q(\lambda)}$.
\item
A hash function $H$ from $\{0,1\}^{s(\lambda)}$ to $\{0,1\}^{r(\lambda)}$ modeled as a quantumly-accessible random oracle.
\end{itemize}
The construction is as follows.
\begin{description}
\item[$\Commit(1^{\lambda},m)$:] $ $
\begin{itemize}
    \item Generate $\mathsf{ske.sk}\leftarrow\mathsf{SKE}.\keygen(1^{\lambda})$, $R\leftarrow\{0,1\}^{s(\lambda)}$, $R'\leftarrow\{0,1\}^{t(\lambda)}$, and a hash function $H$ from $\{0,1\}^{s(\lambda)}$ to $\{0,1\}^{r(\lambda)}$.
    \item Compute $\ske.\ct\leftarrow \SKE.\Enc(\ske.\sk,m)$, $f\leftarrow \algo{Classical}.\Commit(R;R')$, and $h\seteq H(R)\oplus\ske.\sk$.
    \item Output $\cment\seteq (\ske.\ct,f,h)$, $\dment_1\seteq R$, $\dment_2\seteq R'$, and $\ck\seteq\ske.\sk$.
\end{itemize}
\item[$\Verify_1(\cment,\dment_1,\dment_2)$:] $ $
\begin{itemize}
    \item Parse $\cment=(\ske.\ct,f,h)$, $\dment_1= R$, and $\dment_2=R'$.
    \item Output $\top$ if $f=\Classical.\Commit(R;R')$, and output $\bot$ otherwise.
\end{itemize}
\item[$\Verify_2(\cment,\dment_1)$:]  $ $
    \begin{itemize}
        \item Parse $\cment=(\ske.\ct,f,h)$ and $\dment_1=R$.
        \item Compute $\ske.\sk'\seteq H(R)\oplus h$.
        \item Output $m'\leftarrow \SKE.\Dec(\ske.\sk',\ske.\ct)$.
\end{itemize}
\item[$\Delete(\cment)$:] $ $
\begin{itemize}
    \item Parse $\cment=(\ske.\ct,f,h)$. 
    \item Compute $\ske.\cert\leftarrow \SKE.\Delete(\ske.\ct)$.
    \item Output $\cert\seteq\ske.\cert$.
\end{itemize}
\item[$\Cert(\cert,\ck)$:] $ $
\begin{itemize}
    \item Parse $\cert=\ske.\cert$ and $\ck= \ske.\sk$.
    \item Output $\top/\bot\leftarrow \SKE.\Verify(\ske.\sk,\ske.\cert)$.
\end{itemize}
\end{description}

\paragraph{Correctness.}
The decommitment and deletion correctness easily follow from the correctness of $\Sigma_{\mathsf{skcd}}$.

\paragraph{Security.}
We prove the following three theorems.

\begin{theorem}\label{thm:effective_binding}
If $\Sigma_{\cment}$ is perfect binding, then $\Sigma_{\mathsf{ccd}}$ is classical-extractor-based binding.
\end{theorem}

\begin{theorem}\label{thm:everlasting_hiding_2}
If $\Sigma_{\cment}$ is unpredictable and $\Sigma_{\mathsf{skcd}}$ is OT-CD secure, then $\Sigma_{\mathsf{ccd}}$ is certified everlasting hiding.
\end{theorem}

\begin{theorem}\label{thm:computationally_hiding_2}
If $\Sigma_{\cment}$ is unpredictable and $\Sigma_{\mathsf{skcd}}$ is OT-CD secure, then $\Sigma_{\mathsf{ccd}}$ is computationally hiding.
\end{theorem}

\begin{proof}[Proof of Theorem~\ref{thm:effective_binding}]
Due to the perfect binding of $\Sigma_{\cment}=\Classical.\Commit$, there exists a unique $\dment_1^*$ such that $f=\Classical.\Commit(\dment_1^*;\dment_2)$ for a given $f$. 
Let $\mathsf{Ext}$ be the algorithm that finds such $\dment_1^*$ and outputs it. 
(If there is no such $\dment_1^*$, then $\mathsf{Ext}$ outputs $\bot$.)
Then, for any $\cment=(\ske.\ct,f,h)$, any $\dment_1\neq \dment_1^*$, and any $\dment_2$,  
\begin{eqnarray*}
\Pr[\Verify(\cment,\dment=(\dment_1,\dment_2))=\bot]\ge\Pr[f\neq \Classical.\Commit(\dment_1,\dment_2)]=1,
\end{eqnarray*}
which completes the proof.
\end{proof}

\begin{proof}[Proof of Theorem~\ref{thm:everlasting_hiding_2}]
For clarity, we describe how the experiment works against an adversary $\cA\coloneqq (\cA_1,\cA_2)$ consisting of any QPT adversary $\cA_1$ and any quantum unbounded time adversary $\cA_2$.

\begin{description}
\item[$\expa{\Sigma_{\mathsf{ccd}},\cA}{ever}{hide}(\lambda,b)$:] This is the original experiment.

\begin{enumerate}
    \item A uniformly random function $H$ from $\{0,1\}^{s(\lambda)}$ to $\{0,1\}^{r(\lambda)}$ is chosen. $\cA_1$ and $\cA_2$ can make arbitrarily many quantum queries to $H$ at any time in the experiment. 
    \item $\cA_1$ chooses $(m_0,m_1)\leftarrow \cM^2$, and sends $(m_0,m_1)$ to the challenger.
    \item The challenger generates $\ske.\sk\leftarrow\SKE.\keygen(1^{\lambda})$, $R\leftarrow\{0,1\}^{s(\lambda)}$ and $R'\lrun\{0,1\}^{t(\lambda)}$.
    The challenger computes $\ske.\ct\leftarrow\SKE.\Enc(\ske.\sk,m_b)$, $f\coloneqq \Classical.\Commit(R;R')$ and $h\coloneqq H(R)\oplus \ske.\sk$, and sends $(\ske.\ct,f,h)$ to $\cA_1$.
    \item $\cA_1$ sends $\ske.\cert$ to the challenger and sends its internal state to $\cA_2$.
    \item If $\top\lrun\SKE.\Verify(\ske.\sk,\ske.\cert)$, the challenger outputs $\top$ and sends $(R,R',\ske.\sk)$ to $\cA_2$.
    Otherwise, the challenger outputs $\bot$ and sends $\bot$ to $\cA_2$.
    \item $\cA_2$ outputs $b'$.
    \item If the challenger outputs $\top$, then the output of the experiment is $b'$.
    Otherwise, the output of the experiment is $\bot$.
\end{enumerate}
What we have to prove is that  
\begin{align}
    \advb{\Sigma_{\mathsf{ccd}},\cA}{ever}{hide}(\lambda)\seteq
\left|\Pr[\expa{\Sigma_{\mathsf{ccd}},\cA}{ever}{hide}(\lambda,0)=1]-
\Pr[\expa{\Sigma_{\mathsf{ccd}},\cA}{ever}{hide}(\lambda,1)=1]\right|\leq \negl(\lambda).
\end{align}
We define the following sequence of hybrids.

\item[$\sfhyb{1}{}(b)$:]
This is identical to $\expa{\Sigma_{\mathsf{ccd}},\cA}{ever}{hide}(\lambda,b)$ except that the oracle given to $\cA_1$ is replaced with $H_{R\rightarrow H'}$
which is $H$ reprogrammed according to $H'$ on an input $R$ where $H'$ is another independent uniformly random function. More formally, $H_{R\rightarrow H'}$ is defined by 
\begin{align}
    H_{R\rightarrow H'}(R^*) \coloneqq
    \begin{cases}
    H(R^*)~~~&(R^*\neq R)\\
    H'(R^*)~~~&(R^*=R).
    \end{cases}
\end{align}
    We note that the challenger still uses $H$ to generate $h$, and the oracle which $\cA_2$ uses is still $H$ similarly to the original experiment.

\item[$\sfhyb{2}{}(b)$:]
This is identical to $\sfhyb{1}{}(b)$
except for the following three points.
First, the challenger generates $h$ uniformly at random. 
Second, the oracle given to $\cA_1$ is replaced with $H'$ which is an independent uniformly random function.
Third, the oracle given to $\cA_2$ is replaced with $H'_{R\rightarrow h\oplus \ske.\sk}$ which is $H'$ reprogrammed to $h\oplus \ske.\sk$ on an input $R$. 
More formally, $H'_{R\rightarrow h\oplus \ske.\sk}$ is defined by 
\begin{align}
    H'_{R\rightarrow h\oplus \ske.\sk}(R^*)\seteq
    \begin{cases}
    H'(R^*)~~~&(R^*\neq R)\\
    h\oplus \ske.\sk ~~~&(R^*=R).
    \end{cases}
\end{align}

\end{description}

\begin{proposition}\label{prop:exp=hyb_1}
If $\Sigma_{\mathsf{com}}$ is unpredictable, then
\begin{align}
\abs{\Pr[\expa{\Sigma_{\mathsf{ccd}},\cA}{ever}{hide}(\lambda,b)=1]- \Pr[\sfhyb{1}{}(b)=1]}\leq\negl(\lambda).
\end{align}
\end{proposition}

\begin{proof}
The proof is similar to \cite[Propositoin~5.8]{EPRINT:HMNY21}, but note that this time we have to consider
an unbounded adversary after the certificate is issued unlike the case of \cite{EPRINT:HMNY21}.
We assume that $\abs{\Pr[\expa{\Sigma_{\mathsf{ccd}},\cA}{ever}{hide}(\lambda,b)=1]- \Pr[\sfhyb{1}{}(b)=1]}$ is non-negligible,
and construct an adversary $\cB$ that breaks the unpredictability of $\Sigma_{\mathsf{com}}$.
For notational simplicity, we denote $\expa{\Sigma_{\mathsf{ccd}},\cA}{ever}{hide}(\lambda,b)$ by $\sfhyb{0}{}(b)$.
We consider an algorithm $\widetilde{\cA}$ that works as follows. $\widetilde{\cA}$ is given an oracle $\mathcal{O}$,
which is either $H$ or $H_{R\rightarrow H'}$, and an input $z$ that consists of $R$ and the whole truth table of $H$, where $R\leftarrow\{0,1\}^{s(\lambda)}$,
and $H$ and $H'$ are uniformly random functions. $\widetilde{\cA}$ runs $\sfhyb{0}{}(b)$ except that it uses its own oracle $\mathcal{O}$ to simulate $\cA_1$'s random oracle queries.
On the other hand, $\widetilde{\cA}$ uses $H$ to simulate $h$ and $\cA_2$'s random oracle queries regardless of $\mathcal{O}$, which is possible because the truth table of $H$ is included in the input $z$.
By definition, we have 
\begin{align}
    \Pr[\sfhyb{0}{}(b)=1]=\Pr[\widetilde{\cA}^{H}(R,H)=1] 
\end{align}
and 
\begin{align}
    \Pr[\sfhyb{1}{}(b)=1]=\Pr[\widetilde{\cA}^{H_{R\rightarrow H'}}(R,H)=1] 
\end{align}
where $H$ in the input means the truth table of $H$.
We apply the one-way to hiding lemma (\cref{lemma:one-way_to_hiding}) to the above $\widetilde{\cA}$.
Note that $\widetilde{\cA}$ is inefficient, but the one-way to hiding lemma is applicable to inefficient algorithms. Then if we let $\widetilde{\cB}$ be the algorithm that measures uniformly chosen query of $\widetilde{\cA}$,
we have 
\begin{align}
    \abs{\Pr[\widetilde{\cA}^H(R,H)=1]-\Pr[\widetilde{\cA}^{H_{R\rightarrow H'}}(R,H)=1]}\leq 2q\sqrt{\Pr[\widetilde{\cB}^{H_{R\rightarrow H'}}(R,H)=R]}.
\end{align}
By the assumption, the LHS is non-negligible, and thus $\Pr[\widetilde{\cB}^{H_{R\rightarrow H'}}(R,H)=R]$ is non-negligible.

Let $\widetilde{\cB}'$ be the algorithm that is the same as $\widetilde{\cB}$ except that it does not take the truth table of $H$ as input, 
and sets $h$ to be uniformly random string instead of setting $h\seteq H(R)\oplus \ske.\sk$. Then we have
\begin{align}
    \Pr[\widetilde{\cB}^{H_{R\rightarrow H'}}(R,H)=R]=\Pr[\widetilde{\cB}'^{H_{R\rightarrow H'}}(R)=R].
\end{align}
The reason is as follows:
First, $\widetilde{\cB}$ uses the truth table of $H$ only for generating $h\coloneqq H(R)\oplus\ske.\sk$, 
because it halts before $\widetilde{\cB}$ simulates $\cA_2$.
Second, the oracle $H_{R\rightarrow H'}$ reveals no information about $H(R)$,
and thus $h$ can be independently and uniformly random. 

Moreover, for any fixed $R$, when $H$ and $H'$ are uniformly random, $H_{R\rightarrow H'}$ is also a uniformly random function, and therefore we have 
\begin{align}
    \Pr[\widetilde{\cB}'^{H_{R\rightarrow H'}}(R)=R]=\Pr[\widetilde{\cB}'^H(R)=R].
\end{align}
Since $\Pr[\widetilde{\cB}^{H_{R\rightarrow H'}}(R,H)=R]$ is non-negligible, $\Pr[\widetilde{\cB}'^{H}(R)=R]$ is also non-negligible.
Recall that $\widetilde{\cB}'^H$ is an algorithm that simulates $\sfhyb{0}{}(b)$ with the modification that $h$ is set to be uniformly random and measures randomly chosen $\cA_1$'s query.
Then it is straightforward to construct an adversary $\cB$ that breaks the unpredictability of $\Sigma_{\cment}$ by using $\widetilde{\cB}'$.
For clarity, let us give the description of $\cB$ as follows.

$\cB$ is given $\Classical.\Commit(R;R')$ from the challenger of 
$\mathsf{Exp}_{\Sigma_{\mathsf{com}},\cB}^{\mathsf{unpre}}(\lambda)$. 
$\cB$ chooses $i\leftarrow[q]$ and runs $\sfhyb{1}{}(b)$
until $\cA_1$ makes $i$-th random oracle query
or $\cA_1$ sends the internal state to $\cA_2$,
where $\cB$ embeds the problem instance $\Classical.\Commit(R;R')$ into those sent to $\cA_1$ instead of generating it by itself. 
$\cB$ measures the $i$-th random oracle query by $\cA_1$, and outputs the measurement outcome. 
Note that $\cB$ can efficiently simulate the random oracle $H$ by Zhandry's compressed oracle technique~\cite{C:Zhandry19}.
It is clear that the probability that $\cB$ outputs $R$ is exactly $\Pr[\widetilde{\cB}'^{H}(R)=R]$, which is non-negligible.
This contradicts the unpredictability of $\Sigma_{\cment}$.
Therefore $\abs{\Pr[\sfhyb{0}{}(b)=1]-\Pr[\sfhyb{1}{}(b)=1]}$ is negligible.
\end{proof}

\begin{proposition}\label{prop:hyb_1=hyb_2}
$\Pr[\sfhyb{1}{}(b)=1]=\Pr[\sfhyb{2}{}(b)=1]$.
\end{proposition}

\begin{proof}
First, let us remind the difference between $\sfhyb{1}{}(b)$ and $\sfhyb{2}{}(b)$.
In $\sfhyb{1}{}(b)$, $\cA_1$ receives $h=\ske.\sk\oplus H(R)$. Moreover, $\cA_1$ can access to the random oracle $H_{R\rightarrow H'}$,
and $\cA_2$ can access to the random oracle $H$.
On the other hand, in $\sfhyb{2}{}(b)$, $\cA_1$ receives uniformly random $h$.
Moreover, $\cA_1$ can access to the random oracle $H'$ instead of $H_{R\rightarrow H'}$, 
and $\cA_2$ can access to the random oracle $H'_{R\rightarrow h\oplus \ske.\sk}$ instead of $H$.

Let $\Pr[(h,H_{R\rightarrow H'},H)=(r,G,G') \mid \sfhyb{1}{}(b)]$
be the probability
that the adversary $\cA_1$ in $\sfhyb{1}{}(b)$ receives a classical bit string $r$ as $h$, 
random oracle which $\cA_1$ can access to is $G$,
and random oracle which $\cA_2$ can access to is $G'$.  
Similarly, let us define $\Pr[(h,H',H'_{R\rightarrow h\oplus \ske.\sk})=(r,G,G') \mid \sfhyb{2}{}(b)]$
for $\sfhyb{2}{}(b)$.
What we have to show is that the following equation holds for any $(r,G,G')$
\begin{align}
    \Pr[(h,H_{R\rightarrow H'},H)=(r,G,G') \mid \sfhyb{1}{}(b)]=\Pr[(h,H',H'_{R\rightarrow h\oplus \ske.\sk})=(r,G,G') \mid \sfhyb{2}{}(b)].
\end{align}

Since $h=\ske.\sk\oplus H(R)$ in $\sfhyb{1}{}(b)$, $H$ is a uniformly random function, and $h$ in $\sfhyb{2}{}(b)$ is uniformly generated, 
\begin{align}
    \Pr[h=r \mid \sfhyb{1}{}(b)]=\Pr[h=r \mid \sfhyb{2}{}(b)]
\end{align}
holds for any $r$.

For any classical bit string $r$ and any random oracle $G$, we have
\begin{align}
    \Pr[H_{R\rightarrow H'}=G \mid h=r,\sfhyb{1}{}(b)]=\Pr[H'=G \mid h=r,\sfhyb{2}{}(b)].
\end{align}
This is shown as follows.
First, in $\sfhyb{1}{}(b)$, from the construction of $H_{R\rightarrow H'}$, $H_{R\rightarrow H'}(R)$ is independent from $h$ for any $R\in\{0,1\}^{s(\lambda)}$.
Furthermore, since $H$ and $H'$ is random oracle, $H_{R\rightarrow H'}(R)$ is uniformly random for any $R\in \{0,1\}^{s(\lambda)}$.
Second, in $\sfhyb{2}{}(b)$, from the construction of $H'$, $H'(R)$ is independent from $h$ for any $R\in\{0,1\}^{s(\lambda)}$.
Furthermore, since $H'$ is random oracle, $H'(R)$ is uniformly random for any $R\in \{0,1\}^{s(\lambda)}$.
Therefore, we have the above equation.

For any classical bit string $r$ and any random oracles $G$ and $G'$, we have
\begin{align}
    \Pr[H=G' \mid (h,H_{R\rightarrow H'})=(r,G),\sfhyb{1}{}(b)]=\Pr[H'_{R\rightarrow h\oplus \ske.\sk}=G' \mid (h,H')=(r,G),\sfhyb{2}{}(b)].
\end{align}
This can be shown as follows.
First, in $\sfhyb{1}{}(b)$,
we obtain $H(R)=r\oplus \ske.\sk$,
because $h\coloneqq \ske.\sk\oplus H(R)$ and $h=r$. 
Furthermore, from the definition of $H_{R\rightarrow H'}$, we obtain $H(R^*)=G(R^*)$ for $R^*\neq R$.
Second, in $\sfhyb{2}{}(b)$,
from the definition of $H'_{R\rightarrow h\oplus \ske.\sk}$, we have $H'_{R\rightarrow h\oplus \ske.\sk}(R)=r\oplus \ske.\sk$ and
$H'_{R\rightarrow h\oplus \ske.\sk}(R^*)=G(R^*)$ for $R^*\neq R$.

From all above discussions, we have
\begin{align}
    \Pr[(h,H_{R\rightarrow H'},H)=(r,G,G') \mid \sfhyb{1}{}(b)]=\Pr[(h,H',H'_{R\rightarrow h\oplus \ske.\sk})=(r,G,G') \mid \sfhyb{2}{}(b)].
\end{align}
\end{proof}

\begin{proposition}\label{prop:hyb_2}
If $\Sigma_{\mathsf{skcd}}$ is OT-CD secure, then 
\begin{align}
\abs{\Pr[\sfhyb{2}{}(1)=1]-\Pr[\sfhyb{2}{}(0)=1]}\leq\negl(\lambda).
\end{align}
\end{proposition}

\begin{proof}

To show this, we assume that $\abs{\Pr[\sfhyb{2}{}(1)=1]-\Pr[\sfhyb{2}{}(0)=1]}$ is non-negligible,
and construct an adversary $\cB$ that breaks the OT-CD security of $\Sigma_{\mathsf{skcd}}$.

$\cB$ plays the experiment $\expb{\Sigma_{\mathsf{skcd}},\cB}{otsk}{cert}{del}(\lambda,b')$ for some $b'\in\{0,1\}$. 
First, $\cB$ sends $(m_0,m_1)\in \Ms^2$ to the challenger of $\expb{\Sigma_{\mathsf{skcd}},\cB}{otsk}{cert}{del}(\lambda,b')$.
$\cB$ receives $\ske.\ct$ from the challenger of $\expb{\Sigma_{\mathsf{skcd}},\cB}{otsk}{cert}{del}(\lambda,b')$.
$\cB$ generates $R\leftarrow\{0,1\}^{s(\lambda)}$, $R'\leftarrow\{0,1\}^{t(\lambda)}$ and $h\leftarrow\{0,1\}^{r(\lambda)}$, and computes $f\coloneqq\Classical.\Commit(R;R')$.
$\cB$ sends $(\ske.\ct,f,h)$ to $\cA_1$. 
$\cB$ simulates the random oracle $H'$ given to $\cA_1$ by itself.
At some point, $\cA_1$ sends $\ske.\cert$ to $\cB$, and sends the internal state to $\cA_2$.
$\cB$ passes $\ske.\cert$ to the challenger of $\expb{\Sigma_{\mathsf{skcd}},\cB}{otsk}{cert}{del}(\lambda,b')$.

The challenger of
$\expb{\Sigma_{\mathsf{skcd}},\cB}{otsk}{cert}{del}(\lambda,b')$
runs $\SKE.\Verify(\ske.\sk,\ske.\cert)\to\top/\bot$. 
If it is $\top$, 
the challenger sends $\ske.\sk$ to $\cB$.
In that case, $\cB$ outputs $\top$, and sends $(R,R',\ske.\sk)$ to $\cA_2$.
We denote this event by $\mathsf{Reveal_{sk}}(b')$.
$\cB$ simulates $\cA_2$, and outputs the output of $\cA_2$.
On the other hand, if 
$\SKE.\Verify(\ske.\sk,\ske.\cert)\to\bot$, 
then the challenger sends $\bot$ to $\cB$.
In that case,
$\cB$ outputs $\bot$ and aborts.
Note that $\cB$ can simulate the random oracle 
$H'_{R\rightarrow h\oplus \ske.\sk}$ 
given to $\cA_2$ when $\cB$ does not abort,
because $\cB$ receives $\ske.\sk$ from the challenger of $\expb{\Sigma_{\mathsf{skcd}},\cB}{otsk}{cert}{del}(\lambda,b')$ when $\cB$ does not abort.

Now we have
\begin{align}
    &\advc{\Sigma_{\mathsf{skcd}},\cB}{otsk}{cert}{del}(\lambda)\\
    & \coloneqq\abs{\Pr[ \expb{\Sigma_{\mathsf{skcd}},\cB}{otsk}{cert}{del}(\secp, b')=1\mid b'=0] - 
    \Pr[ \expb{\Sigma_{\mathsf{skcd}},\cB}{otsk}{cert}{del}(\secp, b')=1\mid b'=1] }\\
    &=\abs{\Pr[ \cB=1\wedge\mathsf{Reveal_{sk}}(b') \mid b'=0] - 
    \Pr[ \cB=1\wedge\mathsf{Reveal_{sk}}(b') \mid b'=1] }\\
    &=\abs{\Pr[ \cA_{2}=1\wedge\mathsf{Reveal_{sk}}(b') \mid b'=0] - 
    \Pr[ \cA_{2}=1\wedge\mathsf{Reveal_{sk}}(b') \mid b'=1] }\\
    &=|\Pr[\sfhyb{2}{}(b')=1\wedge\mathsf{Reveal_{sk}}(b')\mid b'=0]-\Pr[\sfhyb{2}{}(b')=1\wedge\mathsf{Reveal_{sk}}(b')\mid b'=1]|\\
    &=|\Pr[\sfhyb{2}{}(b')=1\mid b'=0]-\Pr[\sfhyb{2}{}(b')=1\mid b'=1]|\\
    &=|\Pr[\sfhyb{2}{}(0)=1]-\Pr[\sfhyb{2}{}(1)=1]|.
\end{align}
In the second equation, we have used the fact that $\expb{\Sigma_{\mathsf{skcd}},\cB}{otsk}{cert}{del}(\secp, b)=1$ if and only if $\cB=1$ and the challenger of $\expb{\Sigma_{\mathsf{skcd}},\cB}{otsk}{cert}{del}(\lambda,b')$ outputs $\top$. 
In the third equation, we have used the fact that the output of $\cB$ is equal to the output of $\cA_2$ conditioned that $\mathsf{Reveal_{sk}}(b')$ occurs.
In the fourth equation, we have used the fact that $\cB$ simulates the challenger of $\sfhyb{2}{}(b)$ when $\mathsf{Reveal_{sk}}(b)$ occurs.
In the fifth equation, we have used the fact that $\sfhyb{2}{}(b)=1$ only when $\mathsf{Reveal_{sk}}(b)$ occurs.
Since $|\Pr[\sfhyb{2}{}(0)=1]-\Pr[\sfhyb{2}{}(1)=1]|$ is non-negligible,
$\advc{\Sigma_{\mathsf{skcd}},\cB}{otsk}{cert}{del}(\lambda)$ is non-negligible.
This contradicts the OT-CD security of $\Sigma_{\mathsf{skcd}}$.
\end{proof}
By \cref{prop:exp=hyb_1,prop:hyb_1=hyb_2,prop:hyb_2}, we immediately obtain
\cref{thm:everlasting_hiding_2}.

\end{proof}

\ifnum\submission=1
\begin{proof}[Proof of \cref{thm:computationally_hiding_2}]
This can be shown similarly to \cref{thm:everlasting_hiding_2}.
For the convenience of readers, we provide a proof in \cref{Sec:Computational_hiding}.
\end{proof}
\else

\input{Proof_C_hiding}
\fi

%% file: Proof_C_hiding.tex
\begin{proof}[Proof of \cref{thm:computationally_hiding_2}]
For clarity, we describe how the experiment works against a QPT adversary $\cA$.
\begin{description}
\item[$\expa{\Sigma_{\mathsf{ccd}},\cA}{c}{hide}(\lambda,b)$:] This is the original experiment.
\begin{enumerate}
    \item A uniformly random function $H$ from $\{0,1\}^{s(\lambda)}$ to $\{0,1\}^{r(\lambda)}$ is chosen, and $\cA$ can make arbitrarily quantum queries to $H$ at any time in the experiment.
    \item $\cA$ chooses $(m_0,m_1)\leftarrow \cM^2$, and sends $(m_0,m_1)$ to the challenger.
    \item The challenger generates $\ske.\sk\leftarrow\SKE.\keygen(1^{\lambda})$, $R\leftarrow\{0,1\}^{s(\lambda)}$ and $R'\leftarrow\{0,1\}^{t(\lambda)}$.
    The challenger computes $\ske.\ct\leftarrow\SKE.\Enc(\ske.\sk,m_b)$, $f\coloneqq \Classical.\Commit(R;R')$ and $h\coloneqq H(R)\oplus \ske.\sk$, and sends $(\ske.\ct,f,h)$ to $\cA$.
    \item $\cA$ outputs $b'$. The output of the experiment is $b'$.
\end{enumerate}
Note that what we have to prove is 
\begin{align}
    \advb{\Sigma_{\mathsf{ccd}},\cA}{c}{hide}\coloneqq
    \left|\Pr[\expa{\Sigma_{\mathsf{ccd}},\cA}{c}{hide}(\lambda,0)=1]-\Pr[\expa{\Sigma_{\mathsf{ccd}},\cA}{c}{hide}(\lambda,1)=1]\right|\leq \negl(\lambda).
\end{align}
We define the following sequence of hybrids.

\item[$\sfhyb{1}{}(b)$:]
This is identical to $\expa{\Sigma_{\mathsf{ccd}},\cA}{c}{hide}(\lambda,b)$ except that the oracle given to $\cA$ is replaced with $H_{R\rightarrow H'}$
which is $H$ reprogrammed according to $H'$ on an input $R$ where $H'$ is another independent random function.
More formally, $H_{R\rightarrow H'}$ is defined by 
\begin{align}
    H_{R\rightarrow H'} (R^*)\coloneqq
    \begin{cases}
    H(R^*)~~~&(R^*\neq R)\\
    H'(R^*)~~~&(R^*=R).
    \end{cases}
\end{align}
We note that the challenger still uses $H$ to generate $h$.

\item[$\sfhyb{2}{}(b)$:]
This is identical to $\sfhyb{1}{}(b)$ except that the challenger generates $h$ uniformly random.

\end{description}
\begin{proposition}\label{prop:exp=hyb_1_comp}
If $\Sigma_{\cment}$ is unpredictable, then
\begin{align}
\left|\Pr[\expa{\Sigma_{\mathsf{ccd}},\cA}{c}{hide}(\lambda,b)=1]-\Pr[\sfhyb{1}{}(b)=1]\right|\leq\negl(\lambda). 
\end{align}
\end{proposition}

\begin{proof}
It is the same as that of \cref{prop:exp=hyb_1}.
\end{proof}

\begin{proposition}\label{prop:hyb_1=hyb_2_comp}
$\Pr[\sfhyb{1}{}(b)=1]=\Pr[\sfhyb{2}{}(b)=1]$.
\end{proposition}

\begin{proof}
This is similar to the proof of \cref{prop:hyb_1=hyb_2}.
For clarity, we describe the proof.
The difference between $\sfhyb{1}{}(b)$ and $\sfhyb{2}{}(b)$ is as follows.
In $\sfhyb{1}{}(b)$, $\cA$ receives $h\coloneqq H(R)\oplus \ske.\sk$.
In $\sfhyb{2}{}(b)$, $\cA$ receives a uniformly random $h$.
In $\sfhyb{1}{}(b)$, $h$ is uniformly random since $H(R)$ is uniformly distributed.
Therefore, the probability distribution that $\cA$ in $\sfhyb{1}{}(b)$ receives $h$ is equal to the probability distribution that $\cA$ in $\sfhyb{2}{}(b)$ receives $h$.
This completes the proof.
\end{proof}

\begin{proposition}\label{prop:hyb_2_comp}
If $\Sigma_{\mathsf{skcd}}$ is OT-CD secure, then 
\begin{align}
|\Pr[\sfhyb{2}{}(0)=1]-\Pr[\sfhyb{2}{}(1)=1]|\leq\negl(\lambda).
\end{align}
\end{proposition}

\begin{proof}
To show this, we assume that $|\Pr[\sfhyb{2}{}(1)=1]-\Pr[\sfhyb{2}{}(0)=1]|$ is non-negligible,
and construct an adversary $\cB$ that breaks the OT-CD security of $\Sigma_{\mathsf{skcd}}$.

First, $\cB$ sends $(m_0,m_1)\in \Ms^2$ to the challenger of $\expb{\Sigma_{\mathsf{skcd}},\cB}{otsk}{cert}{del}(\lambda,b')$.
$\cB$ receives $\ske.\ct$ from the challenger of $\expb{\Sigma_{\mathsf{skcd}},\cB}{otsk}{cert}{del}(\lambda,b')$ 
and generates $R\leftarrow\{0,1\}^{s(\lambda)}$, $R'\leftarrow\{0,1\}^{t(\lambda)}$, $f\seteq \Classical.\Commit(R;R')$ and $h\leftarrow\{0,1\}^{r(\lambda)}$.
$\cB$ sends $(\ske.\ct,f,h)$ to $\cA$.
$\cB$ simulates the random oracle given to $\cA$.

\begin{itemize}
    \item If $b'=0$, $\cB$ simulates the challenger of $\sfhyb{2}{}(0)$.
    \item If $b'=1$, $\cB$ simulates the challenger of $\sfhyb{2}{}(1)$. 
\end{itemize}
Thus, if $\cA$ distinguishes the two experiments, $\cB$ breaks the OT-CD security of $\Sigma_{\mathsf{skcd}}$
by generating $\ske.\cert$ and sends it to the challenger of $\expb{\Sigma_{\mathsf{skcd}},\cB}{otsk}{cert}{del}(\lambda,b')$.
This completes the proof. 
\end{proof}
By \cref{prop:exp=hyb_1_comp,prop:hyb_1=hyb_2_comp,prop:hyb_2_comp}, we immediately obtain \cref{thm:computationally_hiding_2}.  
\end{proof}

%% file: Sec_Zeroknowledge.tex
\section{Certified Everlasting Zero-Knowledge Proof for QMA}\label{sec:ZK}

In this section, we define and construct the certified everlasting zero-knowledge proof for $\QMA$.
In \cref{sec:def_everlasting_ZK}, we define the certified everlasting zero-knowledge proof for $\QMA$.
We then construct a three round protocol with completeness-soundness gap $\frac{1}{\poly(\lambda)}$ in \cref{sec:3_construction},
and finally amplify the gap to $1-\negl(\lambda)$ with the sequential repetition in \cref{sec:Sequential_repetition}.

\subsection{Definition}\label{sec:def_everlasting_ZK}
We first define a quantum interactive protocol. 
Usually, in zero-knowledge proofs or arguments, we do not consider prover's output. However, in this paper,
we also consider prover's output,
because we are interested in the certified everlasting zero-knowledge.
Furthermore, in this paper, we consider only an interactive proof, which means that a malicious prover is unbounded. 

\begin{definition}[Quantum Interactive Protocol]\label{def:quantum_interactive_protocol}
A quantum interactive protocol is modeled as an interaction between QPT machines $\cP$ referred as a prover and $\cV$ referred as a verifier.  
We denote by $\langle \cP(x_P),\cV(x_V)\rangle (x)$ an execution of the protocol where $x$ is a common input, 
$x_P$ is $\cP$'s private input, and $x_V$ is $\cV$'s private input.
We denote by $\mathrm{OUT}_{\cV}\langle \cP(x_P),\cV(x_V)\rangle(x)$ the final output of $\cV$ in the execution.
An honest verifier's output is $\top$ indicating acceptance or $\bot$ indicating rejection, and a malicious verifier's output is an arbitrary quantum state. 
We denote by $\mathrm{OUT}_{\cP}\langle \cP(x_P),\cV(x_V)\rangle(x)$ the final output of $\cP$ in the execution.
An honest prover's output is $\top$ indicating acceptance or $\bot$ indicating rejection. 
We also define $\mathrm{OUT'}_{\cP,\cV}\langle \cP(x_{P}),\cV(x_{V})\rangle(x)$ by
\begin{align}
    \mathrm{OUT'}_{\cP,\cV}\langle \cP(x_P),\cV(x_V)\rangle(x)
    \coloneqq
    \begin{cases}
    \left(\top,\mathrm{OUT}_{\cV}\langle \cP(x_P),\cV(x_V)\rangle(x)\right)~~~&(\mathrm{OUT}_{\cP}\langle \cP(x_P),\cV(x_V)\rangle(x)=\top)\\
    (\bot,\bot)~~~&(\mathrm{OUT}_{\cP}\langle \cP(x_P),\cV(x_V)\rangle(x)\neq\top).
    \end{cases}
\end{align}
\end{definition}

We next define a computational zero-knowledge proof for $\QMA$, which is the standard definition.

\begin{definition}[Computational Zero-Knowledge Proof for QMA]\label{def:ZK}
A $c$-complete $s$-sound computational zero-knowledge proof for a $\QMA$ 
promise problem $A=(A_{\yes},A_{\no})$ is a quantum interactive protocol between a QPT prover $\cP$ and a QPT verifier $\cV$ that satisfies the followings:
\begin{description}
\item[$c$-completeness:]
For any $\statement\in A_{\yes}$ and any $\witness\in R_{A}(\statement)$, 
\begin{align}
\Pr[\mathsf{Out}_{\cV}\langle \cP(\witness^{\otimes k(|\statement|)}), \cV\rangle (\statement)=\top ]\ge c
\end{align}
for some polynomial $k$.
\item[$s$-soundness:]
 For any $\statement\in A_{\no}$ and
 any unbounded-time prover $\cP^*$,
 \begin{align}
    \Pr[\mathsf{Out}_{\cV}\langle \cP^*, \cV\rangle (\statement)=\top]\leq s.
 \end{align}

\item[Computational zero-knowledge:]
There exists a QPT algorithm $\cS$ such that 
\begin{align}
\mathsf{OUT}_{\cV^*}\langle \cP(\witness^{\otimes k(|\statement|)}),\cV^*(\cdot) \rangle(\statement) \approx_c \cS(\statement,\cV^*,\cdot~)
\end{align}
for any QPT malicious verifier $\cV^*$,
any $\statement\in A_{\yes}\cap \bit^\lambda$, any $\witness\in R_{A}(\statement)$, and some polynomial $k$.
Note that $\mathsf{OUT}_{\cV^*}\langle \cP(\witness^{\otimes k(|\statement|)}),\cV^*(\cdot) \rangle(\statement)$ and $\cS(\statement,\cV^*,\cdot~)$ are quantum channels
that map any quantum state $\xi$ to quantum states $\mathsf{OUT}_{\cV^*}\langle \cP(\witness^{\otimes k(|\statement|)}),\cV^*(\xi) \rangle(\statement)$ and $\cS(\statement,\cV^*,\xi)$, respectively.

\end{description}
We just call it a computational zero-knowledge proof if it satisfies $(1-\negl(|\statement|))$-completeness, 
$\negl(|\statement|)$-soundness, and 
computational zero-knowledge.
\end{definition}

We finally define a certified everlasting zero-knowledge proof for $\QMA$, which is the main target of
this paper.

\begin{definition}[Certified Everlasting Zero-Knowledge Proof for QMA]\label{def:Everlasting_zeroknowledge}
A certified everlasting zero-knowledge proof for a $\QMA$ promise problem $A=(A_{\yes},A_{\no})$ is 
a computational zero-knowledge proof for $A$ (\cref{def:ZK}) that additionally satisfies the followings:
\begin{description}
\item[Prover's completeness:]
$\Pr[\mathrm{OUT}_{\cP}\langle \cP(\witness^{\otimes k(|\statement|)}),\cV\rangle(\statement)=\top]\ge1-\negl(\lambda)$
for any $\statement\in A_{\yes}\cap\bit^\lambda$ and any $\witness\in R_A(\statement)$.

\item[Certified everlasting zero-knowledge:]
There exists a QPT algorithm $\cS$ such that 
\begin{align}
    \mathrm{OUT'}_{\cP,\cV^*}\langle \cP(\witness^{\otimes k(|\statement|)}),\cV^*(\cdot)\rangle (\statement) \approx_{s}
    \cS(\statement,\cV^*,\cdot~)
\end{align}
for any QPT malicious verifier $\cV^*$,
any $\statement\in A_{\yes}\cap \bit^\lambda$, any $\witness\in R_{A}(\statement)$, and some polynomial $k$. 
Note that $\mathsf{OUT'}_{\cP,\cV^*}\langle \cP(\witness^{\otimes k(|\statement|)}),\cV^*(\cdot) \rangle(\statement)$ and $\cS(\statement,\cV^*,\cdot~)$ are quantum channels
that map any quantum state $\xi$ to quantum states
$\mathsf{OUT'}_{\cP,\cV^*}\langle \cP(\witness^{\otimes k(|\statement|)}),\cV^*(\xi) \rangle(\statement)$ and $\cS(\statement,\cV^*,\xi)$, respectively.

\end{description}
\end{definition}
\begin{remark}
We remark that certified everlasting zero-knowledge does not imply computational zero-knowledge since it does not require anything if the prover does not output $\top$.
\end{remark}

\subsection{Construction of Three Round Protocol}\label{sec:3_construction}
In this section, we construct a three round protocol with completeness-soundness gap $\frac{1}{\poly(\lambda)}$.  
In the next section, we will amplify its completeness-soundness gap by the sequential repetition.

In the following, $n$, $m$, $\Pi_{c}$, $\rho_\hist$, and $\rho_{\Sim}^{\statement,S}$ are given in \cref{def:k-SimQMA}. 
Let $S_c\subseteq[n]$ be the set of qubits on which $\Pi_c$ acts non-trivially.
The three round protocol $\Sigma_{\Xi\mathsf{cd}}$ is constructed from commitment with certified everlasting hiding and classical-extractor-based binding, $\Sigma_{\mathsf{ccd}}=(\Commit,\Verify,\Delete,\Cert)$.

\begin{description}
    \item[The first action by the prover (commitment phase):] $ $
    \begin{itemize}
        \item Generate $x,z\leftarrow\{0,1\}^n$.
        \item Compute 
        \begin{align}
       & (\cment_i(x_i),\dment_i(x_i),\ck_i(x_i))\leftarrow\Commit(1^\lambda,x_i)\\
       & (\cment_i(z_i),\dment_i(z_i),\ck_i(z_i))\leftarrow \Commit(1^\lambda,z_i)
        \end{align}
        for all $i\in[n]$. 
        \item Generate a simulatable witness $\rho_\hist$ for the instance $\statement$ and generate $X^xZ^z\rho_\hist Z^zX^x$.
        \item Send the first message (commitment),
        $\msg_1\seteq(X^xZ^z\rho_\hist Z^zX^x)\otimes \cment(x)\otimes \cment(z)$, to the verifier,
        where $\cment(x)\seteq\bigotimes_{i=1}^{n} \cment_i(x_i)$ and $\cment(z)\seteq\bigotimes_{i=1}^n \cment_i(z_i)$. 
    \end{itemize}
    \item[The second action by the verifier (challenge phase):] $ $
    \begin{itemize} 
        \item Generate $c\leftarrow [m]$.
        \item Compute $\cert_i(x_i)\leftarrow\Delete(\cment_i(x_i))$ and $\cert_i(z_i)\leftarrow \Delete(\cment_i(z_i))$ for all $i\in \overline{S}_c$. 
        \item Send the second message (challenge), $\msg_2\seteq (c,\{\cert_i(x_i),\cert_i(z_i)\}_{i\in\overline{S}_c})$, to the prover.
    \end{itemize}
    \item[The third action by the prover (reply phase):] $ $
    \begin{itemize}
        \item Send the third message (reply), $\msg_3\seteq\{\dment_i(x_i),\dment_i(z_i)\}_{i\in S_{c}}$, to the verifier.
        \item Output $\top$ if $\top\leftarrow\Cert(\cert_i(x_i),\ck_i(x_i))$ and $\top\leftarrow\Cert(\cert_i(z_i),\ck_i(z_i))$ for all $i\in\overline{S}_c$,
         and output $\bot$ otherwise.
     \end{itemize}
     \item[The fourth action by the verifier (verification phase):] $ $
     \begin{itemize}
        \item Compute $x'_i\leftarrow\Verify(\cment_i(x_i),\dment_i(x_i))$ and $z'_i\leftarrow\Verify(\cment_i(z_i),\dment_i(z_i))$ for all $i\in S_c$. If $x_i'=\bot$ or $z_i'=\bot$ for at least one $i\in S_c$, output $\bot$ and abort.
        \item Apply $X_i^{x'_i}Z_{i}^{z'_i}$ on the $i$-th qubit of $X^xZ^z\rho_{\hist}Z^zX^x$ for each $i\in S_c$, and
        perform the POVM measurement $\{\Pi_c,I-\Pi_c\}$ on the state. 
        \item Output $\top$ if the result $\Pi_c$ is obtained, and output $\bot$ otherwise.
     \end{itemize}
\end{description}

\begin{theorem}\label{thm:everlasting_zero_proof}
$\Sigma_{\Xi\mathsf{cd}}$ is a certified everlasting zero-knowledge proof for $\QMA$ with $(1-\negl(\lambda))$-completeness and $\left(1-\frac{1}{\poly(\lambda)}\right)$-soundness.
\end{theorem}

This is shown from the following \cref{lem:completeness,lem:soundness,lem:everlasting_zero,lem:computational_zero}.

\begin{lemma}\label{lem:completeness}
$\Sigma_{\Xi\mathsf{cd}}$ satisfies the $(1-\negl(\lambda))$-completeness and prover's completeness.
\end{lemma}

\begin{lemma}\label{lem:soundness}
If $\Sigma_{\mathsf{ccd}}$ is classical-extractor-based binding, then $\Sigma_{\Xi\mathsf{cd}}$ satisfies $\left(1-\frac{1}{\poly(\lambda)}\right)$-soundness.
\end{lemma}

\begin{lemma}\label{lem:everlasting_zero}
If $\Sigma_{\mathsf{ccd}}$ is certified everlasting hiding and computational hiding, then $\Sigma_{\Xi\mathsf{cd}}$ satisfies certified everlasting zero-knowledge.
\end{lemma}

\begin{lemma}\label{lem:computational_zero}
If $\Sigma_{\mathsf{ccd}}$ is computational hiding, then $\Sigma_{\Xi\mathsf{cd}}$ satisfies computational zero-knowledge.
\end{lemma}

\begin{proof}[Proof of \cref{lem:completeness}]
It is clear from the definition of $k$-$\compclass{SimQMA}$ (\cref{def:k-SimQMA}) and the correctness of $\Sigma_{\mathsf{ccd}}$. 
\end{proof}

\begin{proof}[Proof of \cref{lem:soundness}]
Let us show the soundness by analyzing the case for $\statement\in A_{\no}$. 
The prover sends the first message to the verifier.
The first message consists of three registers, $RS$, $RCX$, and $RCZ$.
The register $RCX$ further consists of $n$ registers $\{RCX_i\}_{i\in[n]}$.
The register $RCZ$ also consists of $n$ registers $\{RCZ_i\}_{i\in[n]}$.
If the prover is honest, $RS$ contains $X^xZ^z \rho_\hist Z^zX^x$, $RCX_i$ contains $\cment_i(x_i)$, and
$RCZ_i$ contains $\cment_i(z_i)$.
Let $\cment_{i,x}'$ and $\cment_{i,z}'$ be the (reduced) states of the registers $RCX_i$ and $RCZ_i$, respectively.
Let $f_{i,x}'$ and $f_{i,z}'$ be classical parts of $\cment_{i,x}'$ and $\cment_{i,z}'$, respectively.

The verifier generates $c\leftarrow[m]$, and issues the deletion certificate.
The verifier sends $c$ and the deletion certificate to the prover.
The verifier then receives $\{\dment_1^{x,i},\dment_2^{x,i},\dment_1^{z,i},\dment_2^{z,i}\}_{i\in S_c}$ from the prover. 
For each $i\in[n]$, let us define $\dment_1^{*,x,i}$ and $\dment_1^{*,z,i}$ by
    $\dment_1^{*,x,i}\leftarrow \mathsf{Ext}(f_{i,x}')$ and
    $\dment_1^{*,z,i}\leftarrow \mathsf{Ext}(f_{i,z}')$, respectively. 
Note that each $\dment_1^{*,x,i}$ and $\dment_1^{*,z,i}$ is independent of $c$, because
$\cment_{i,x}'$ and $\cment_{i,z}'$ are sent to the verifier before the verifier chooses $c$.

We have only to consider the case when
$\dment_1^{x,i}=\dment_1^{*,x,i}$
and
$\dment_1^{z,i}=\dment_1^{*,z,i}$
for all $i\in S_c$,
because of the following reason:
Due to the classical-extractor-based binding of $\Sigma_{\mathsf{ccd}}$, 
$\Verify(\cment_{i,x}',(\dment_1^{x,i},\dment_2^{x,i}))=\bot$
for any $\dment_1^{x,i}\neq \dment_1^{*,x,i}$ and any $\dment_2^{x,i}$.
Similarly, 
$\Verify(\cment_{i,z}',(\dment_1^{z,i},\dment_2^{z,i}))=\bot$
for any $\dment_1^{z,i}\neq \dment_1^{*,z,i}$ and any $\dment_2^{z,i}$. 
Therefore, the prover who wants to make the verifier accept has to send
$\dment_1^{x,i}=\dment_1^{*,x,i}$
and
$\dment_1^{z,i}=\dment_1^{*,z,i}$ for all $i\in S_c$.

Let us define
\begin{align}
    p(x,z)\seteq \Pr\left[\bigwedge_{i\in[n]}\left(\Verify_2 (\cment_{i,x}',\dment_1^{*,x,i})\to x_i \wedge \Verify_2(\cment_{i,z}',\dment_1^{*,z,i})\to z_i\right)\right].
\end{align}
Note that $p(x,z)$ is independent of $c$, because
$\{\cment_{i,x}',\cment_{i,z}'\}_{i\in[n]}$
and
$\{\dment_1^{*,x,i},\dment_1^{*,z,i}\}_{i\in [n]}$ are independent of $c$.
Let $\psi$ be the (reduced) state of the register $RS$.
The verifier's acceptance probability is 
\begin{align}
    &\frac{1}{m}\sum_{c\in[m]}\sum_{x,z\in\bit^n}p(x,z)\Tr\left[\Pi_c 
    \left(\prod_{i\in S_c}Z_i^{z_i}X_i^{x_i}\right)
    \psi 
    \left(\prod_{i\in S_c}X_i^{x_i}Z_i^{z_i}\right)\right]\\
     &=\frac{1}{m}\sum_{c\in[m]}\sum_{x,z\in\bit^n}p(x,z)\Tr\left[\Pi_c 
    \left(\prod_{i\in [n]}Z_i^{z_i}X_i^{x_i}\right)
    \psi 
    \left(\prod_{i\in [n]}X_i^{x_i}Z_i^{z_i}\right)\right]\\
      &=\frac{1}{m}\sum_{c\in[m]}
      \Tr\left[\Pi_c 
      \sum_{x,z\in\bit^n}p(x,z)
    \left(\prod_{i\in [n]}Z_i^{z_i}X_i^{x_i}\right)
    \psi 
    \left(\prod_{i\in [n]}X_i^{x_i}Z_i^{z_i}\right)\right]\\
    &\leq 1-\frac{1}{\poly(\lambda)},
\end{align}
where the last inequality comes from \cref{def:k-SimQMA}. 
This completes the proof.
\end{proof}

\begin{proof}[Proof of \cref{lem:everlasting_zero}]
Let us show certified everlasting zero-knowledge.
For a subset $S_c\subseteq[n]$ and $x,z\in\bit^n$, let us define $x^{S_c}\seteq (x^{S_c}_1,x^{S_c}_2,\cdots, x^{S_c}_n)$ and $z^{S_c}\seteq (z^{S_c}_1,z^{S_c}_2,\cdots,z^{S_c}_n)$,
where $x^{S_c}_{i}=x_{i}$ and $z^{S_c}_{i}=z_{i}$ for $i\in S_{c}$,
and $x^{S_c}_{i}=z^{S_c}_{i}=0$ for $i\notin S_{c}$. 

For clarity, we describe how the interactive algorithm $\langle \cP(\witness^{\otimes k(|\statement|)}),\cV^*(\xi)\rangle (\statement)$ runs against a QPT verifier $\cV^*$ with an input $\xi$,
where $\witness$ is the witness and $\statement$ is the instance.
\begin{description}
\item[$\langle \cP(\witness^{\otimes k(|\statement|)}),\cV^*(\xi)\rangle (\statement)$:] $ $
\begin{enumerate}
    \item $\cP$ generates $x,z\leftarrow \{0,1\}^n$, and computes 
    \begin{align}
    &(\cment_i(x_i),\dment_i(x_i),\ck_i(x_i))\leftarrow\Commit(1^\lambda,x_i) \\ 
    &(\cment_i(z_i),\dment_i(z_i),\ck_i(z_i))\leftarrow\Commit(1^\lambda,z_i)
    \end{align}
    for all $i\in [n]$.
    $\cP$ sends $\msg_1\seteq(X^xZ^z\rho_\hist Z^zX^x)\otimes \cment(x)\otimes\cment(z)$ to $\cV^*$. 
    \item $\cV^*$ appends $\xi$ to the received state, and
   runs a QPT circuit $V_1^*$ on it to obtain $(c,\{\cert_{i,x}',\cert_{i,z}'\}_{i\in\overline{S}_c})$.
   $\cV^*$ sends $\msg_2\seteq(c,\{\cert_{i,x}',\cert_{i,z}'\}_{i\in\overline{S}_c})$ to $\cP$.
    \item $\cP$ sends $\msg_3\seteq\{\dment_i(x_i),\dment_i(z_i)\}_{i\in S_c}$ to $\cV^*$.
    \item $\cV^*$ appends $\msg_3$ to its state, and runs a QPT circuit $V_2^*$ on it. $\cV^*$ outputs its state $\xi'$.
    \item $\cP$ computes $\Cert(\cert_{i,x}',\ck_i(x_i))$ and $\Cert(\cert_{i,z}',\ck_i(z_i))$ for all $i\in\overline{S}_c$. 
    If all outputs are $\top$, then $\cP$ outputs $\top$.
    Otherwise, $\cP$ outputs $\bot$.
\end{enumerate}
\end{description}

Next let us define a simulator $\cS^{(1)}$ as follows.

\begin{description}
\item [The simulator $\cS^{(1)}(\statement,\cV^*,\xi)$:] $ $
\begin{enumerate}
    \item Pick $c\leftarrow [m]$ and $x,z\leftarrow \{0,1\}^n$. 
    Compute 
    \begin{align}
   & (\cment_i(x^{S_c}_i),\dment_i(x^{S_c}_i),\ck_i(x^{S_c}_i))\leftarrow\Commit(1^\lambda,x^{S_c}_i)\\ 
   & (\cment_i(z^{S_c}_i),\dment_i(z^{S_c}_i),\ck_i(z^{S_c}_i))\leftarrow\Commit(1^\lambda,z^{S_c}_i)
    \end{align}
    for all $i\in [n]$.
    \item Generate $(X^xZ^z\sigma(c) Z^zX^x) \otimes \cment(x^{S_c})\otimes \cment(z^{S_c})\otimes \xi$,
    where $\sigma(c)\seteq \rho_\simulator^{\statement,S_c}\otimes\left(\prod_{i\in \overline{S}_c}|0\rangle\langle0|_i\right)$.
    Run $V_1^*$ on the state to obtain $(c',\{\cert_{i,x}',\cert_{i,z}'\}_{i\in\overline{S}_{c'}})$.
    
    \item If $c'\neq c$, abort and output a fixed state $\eta$ and the flag state $\mathsf{fail}$. 
    \item Append $\{\dment_i(x^{S_c}_i),\dment_i(z^{S_c}_i)\}_{i\in S_c}$ to its quantum state, and run $V_2^*$ on the state to obtain $\xi'$.
    \item Compute $\Cert(\cert_{i,x}',\ck_i(x^{S_c}_i))$ and $\Cert(\cert_{i,z}',\ck_i(z^{S_c}_i))$ for all $i\in \overline{S}_c$.
    If all outputs are $\top$, then output the state $(\top,\xi')$. Otherwise, output $(\bot,\bot)$.
    Also output the flag state $\mathsf{success}$.
\end{enumerate}
\end{description}

Let us also define other two simulators, $\cS^{(2)}$ and $\cS^{(3)}$, as follows.
\begin{description}
\item [The simulator $\cS^{(2)}(\statement,\witness^{\otimes k(|\statement|)},\cV^*,\xi)$:] $ $
   It is the same as $\cS^{(1)}$ except that
   $\sigma(c)$ is replaced with $\rho_{\hist}$. 
\end{description}

\begin{description}
\item [The simulator $\cS^{(3)}(\statement,\witness^{\otimes k(|\statement|)},\cV^*,\xi)$:] $ $
$\cS^{(3)}(\statement,\witness^{\otimes k(|\statement|)},\cV^*,\cdot~)$ is the channel that postselects the output of
\\$\cS^{(2)}(\statement,\witness^{\otimes k(|\statement|)},\cV^*,\cdot~)$ on the non-aborting state. More precisely, if we
write $\cS^{(2)}(\statement,\witness^{\otimes k(|\statement|)},\cV^*,\rho_{in})=p \rho_{out}\otimes \mathsf{success}+(1-p)\eta\otimes \mathsf{fail}$, 
where $p$ is the non-aborting probability,
$\cS^{(3)}(\statement,\witness^{\otimes k(|\statement|)},\cV^*,\rho_{in})= \rho_{out}$.
\end{description}

\cref{lem:everlasting_zero} is shown from the following \cref{prop:abort,prop:indistinguishable_S_S'_ever,prop:indistinguishable_S'_V^*_ever} 
(whose proofs will be given later) and quantum rewinding lemma (\cref{lemma:rewinding}),
which is used 
to reduce the probability that $\cS^{(1)}$ aborts to $\negl(\lambda)$.
In fact, from \cref{prop:abort,lemma:rewinding}, there exists a quantum circuit $\cS^{(0)}$ of size at most
$O(m\,{\rm poly}(n){\rm size}(\cS^{(1)}))$
such that the probability that $\cS^{(0)}$ aborts is $\negl(\lambda)$, and
the output quantum states of $\cS^{(0)}$ and $\cS^{(1)}$ are $\negl(\lambda)$-close when they do not abort.
From \cref{prop:indistinguishable_S_S'_ever,prop:indistinguishable_S'_V^*_ever}, 
$\cS^{(0)}$ is $\negl(\lambda)$-close to the real protocol,
which completes the proof.
\end{proof}

\begin{proposition}\label{prop:abort}
If $\Sigma_{\mathsf{ccd}}$ is computationally hiding, then
the probability that $\cS^{(1)}$ does not abort is $\frac{1}{m}\pm\negl(\lambda)$.
\end{proposition}


\begin{proposition}\label{prop:indistinguishable_S_S'_ever}
$\cS^{(1)}(\statement,\cV^*,\cdot~)\approx_s\cS^{(2)}(\statement,\witness^{\otimes k(|\statement|)},\cV^*,\cdot~)$
for any $\statement\in A_{\yes}\cap\bit^\lambda$
and any $\witness\in R_A(\statement)$.
\end{proposition}

\begin{proposition}\label{prop:indistinguishable_S'_V^*_ever}
If $\Sigma_{\mathsf{ccd}}$ is certified everlasting hiding,
$\cS^{(3)}(\statement,\witness^{\otimes k(|\statement|)},\cV^*,\cdot~)\approx_s 
\mathsf{OUT'}_{\cP,\cV^*}\langle \cP(\witness^{\otimes k(|\statement|)}),\cV^*(\cdot)\rangle (\statement)$.
\end{proposition}

\begin{proof}[Proof of \cref{prop:abort}]
This can be shown similarly to \cite[Lemma~5.6]{FOCS:BroGri20}. For the convenience of readers, we provide a proof in~\cref{proof:abort}.
\end{proof}

\begin{proof}[Proof of \cref{prop:indistinguishable_S_S'_ever}]
It is clear from the local simulatability (\cref{def:k-SimQMA}) and
the definition of $x^{S_c}$ and $z^{S_c}$ (all $x_i^{S_c}$ and $z_i^{S_c}$ are 0 except for those in $i\in S_c$).
\end{proof}

\begin{proof}[Proof of \cref{prop:indistinguishable_S'_V^*_ever}]
We prove the proposition by contradiction.
We construct an adversary $\cB$ that breaks the security of the certified
everlasting hiding of $\Sigma_{\mathsf{ccd}}$ by assuming 
the existence of a distinguisher $\cD$ that distinguishes two states $\delta_0$ and $\delta_1$,
\begin{align}
&\delta_0\seteq(\mathsf{OUT'}_{\cP,\cV^*}\langle \cP(\witness^{\otimes k(|\statement|)}),\cV^*(\cdot)\rangle (\statement)\otimes I)\sigma\\
&\delta_1\seteq(\cS^{(3)}(\statement,\witness^{\otimes k(|\statement|)},\cV^*,\cdot~)\otimes I)\sigma, 
\end{align}
with a certain state $\sigma$.
Let us describe how $\cB$ works.
\begin{enumerate}
    \item $\cB$ generates $c\leftarrow[m]$ and $x,z\lrun\{0,1\}^n$.
    \item $\cB$ sends $m_0\coloneqq \{x_i,z_i\}_{i\in\overline{S}_c}$ and $m_1\coloneqq0^{2n-10}$ to the challenger of $\expb{\Sigma_{\mathsf{ccd}},\cB}{bit}{ever}{hide}(\lambda,b)$.
    $\cB$ receives commitments from the challenger which is either $\{\cment_i(x_i),\cment_i(z_i)\}_{i\in \overline{S}_c}$ or $\{\cment_i(0),\cment_i(0)\}_{i\in \overline{S}_c}$.
    \item $\cB$ computes 
    \begin{align}
   & (\cment_i(x_i),\dment_i(x_i),\ck_i(x_i))\leftarrow\Commit(1^\lambda,x_i)\\
   & (\cment_i(z_i),\dment_i(z_i),\ck_i(z_i))\leftarrow\Commit(1^\lambda,z_i)
    \end{align}
    for $i\in S_c$ by itself.
    \item $\cB$ generates $X^xZ^z\rho_{\hist}Z^zX^x$. 
    $\cB$ appends commitments and $\sigma$ to the quantum state.
    If the commitments for ${i\in \overline{S}_c}$ are $\{\cment_i(x_i),\cment_i(z_i)\}_{i\in \overline{S}_c}$,
    $\cB$ obtains $(X^xZ^z\rho_{\hist}Z^zX^x)\otimes \cment(x)\otimes\cment(z)\otimes \sigma$.
    If the commitments for ${i\in \overline{S}_c}$ are $\{\cment_i(0),\cment_i(0)\}_{i\in \overline{S}_c}$,
    $\cB$ obtains $(X^xZ^z\rho_{\hist}Z^zX^x)\otimes \cment(x^{S_c})\otimes\cment(z^{S_c})\otimes \sigma$.
    \item $\cB$ runs $V_1^*$ on it to obtain $(c',\{\cert_{i,x}',\cert_{i,z}'\}_{i\in \overline{S}_{c'}})$. 
    $\cB$ aborts when $c\neq c'$.
    \item $\cB$ appends $\{\dment_{i}(x_i),\dment_{i}(z_i)\}_{i\in S_c}$ to the post-measurement state and runs $V_2^*$ on it to obtain $\sigma'$.
    \item $\cB$ sends $\{\cert_{i,x}',\cert_{i,z}'\}_{i\in \overline{S}_{c}}$ to the challenger of $\expb{\Sigma_{\mathsf{ccd}},\cB}{bit}{ever}{hide}(\lambda,b)$,
    and receives $\bot$ or 
    $\{\dment_{i}(x_i),\dment_{i}(z_i)\}_{i\in \overline{S}_c}$ and 
    $\{\ck_i(x_i),\ck_i(z_i)\}_{i\in \overline{S}_c}$ from the challenger.
    \item $\cB$ passes $(\bot,\bot)$ to $\cD$ if $\cB$ receives $\bot$ from the challenger, and passes $(\top,\sigma')$ to $\cD$ otherwise.
    \item When $\cD$ outputs $b$, $\cB$ outputs $b$.
\end{enumerate}

When $\cB$ receives $\{\cment_i(x_i),\cment_i(z_i)\}_{i\in \overline{S}_c}$ from the challenger and it does not abort, 
it simulates \\${\mathsf{OUT}}'_{\cP,\cV^*}\langle \cP(\witness^{\otimes k(|\statement|)}),\cV^*(\cdot)\rangle(\statement)$.
Because $(X^xZ^z\rho_{\hist}Z^zX^x)\otimes \cment(x)\otimes\cment(z)\otimes \sigma$ is independent of $c$,
the probability that $\cB$ does not abort is $\frac{1}{m}$.
Therefore,
$\cB$ can simulate
${\mathsf{OUT}}'_{\cP,\cV^*}\langle \cP(\witness^{\otimes k(|\statement|)}),\cV^*(\cdot)\rangle(\statement)$
with probability $\frac{1}{m}$.

When $\cB$ receives $\{\cment_i(0),\cment_i(0)\}_{i\in \overline{S}_c}$ from the challenger and it does not abort, 
it simulates $\cS^{(3)}(\statement,\witness^{\otimes k(|\statement|)},\cV^*,\cdot~)$.
The probability that
$\cB$ does not abort is $\frac{1}{m}\pm\negl(\lambda)$ from
\cref{prop:abort,prop:indistinguishable_S_S'_ever}. 
Therefore, $\cB$ can simulate
$\cS^{(3)}(\statement,\witness^{\otimes k(|\statement|)},\cV^*,\cdot~)$ with probability
$\frac{1}{m}\pm\negl(\lambda)$.

Therefore, if there exists a distinguisher $\cD$ that distinguishes 
$\delta_0$ and $\delta_1$,
$\cB$ can distinguish $\{\cment_i(x_i),\cment_i(z_i)\}_{i\in \overline{S}_c}$ from $\{\cment_i(0),\cment_i(0)\}_{i\in \overline{S}_c}$.
From \cref{lemma:eachbit},
this contradicts the certified everlasting hiding of $\Sigma_{\mathsf{ccd}}$.

\end{proof}

\begin{proof}[Proof of \cref{lem:computational_zero}]
Computational zero-knowledge can be proven similarly to \cite[Lemma~5.3]{FOCS:BroGri20} because our protocol is identical to theirs if we ignore the deletion certificates, which are irrelevant to the computational zero-knowledge property. For the convenience of readers, we provide a proof in \cref{proof:computational_zero}.
\end{proof}

\subsection{Sequential Repetition for Certified Everlasting Zero-Knowledge Proof for QMA}\label{sec:Sequential_repetition}
In this section, we amplify the completeness-soundness gap of the three-round protocol constructed in the previous
section by sequential repetition.

\begin{theorem}\label{thm:sequential_everlasting_zero_proof}
Let $\Sigma_{\Xi\mathsf{cd}}$ be a certified everlasting zero-knowledge proof for a $\QMA$ promise problem $A$ with $\left(1-\negl(\lambda)\right)$-completeness
and $\left(1-\frac{1}{\rm poly(\lambda)}\right)$-soundness.
For any polynomial $N={\rm poly(\lambda)}$, let $\Sigma_{\Xi\mathsf{cd}}^N$ be the $N$-sequential repetition of $\Sigma_{\Xi\mathsf{cd}}$.
That is, $\cP$ and $\cV$ in $\Sigma_{\Xi\mathsf{cd}}^N$ run $\Sigma_{\Xi\mathsf{cd}}$ sequentially $N$ times.
Let $\cP_j$ and $\cV_j$ be the prover and the verifier in the $j$-th run of $\Sigma_{\Xi\mathsf{cd}}$,
respectively.
$\cP$ in $\Sigma_{\Xi\mathsf{cd}}^N$ outputs $\top$ if $\cP_j$ outputs $\top$ for all $j\in[N]$, and outputs $\bot$ otherwise.
$\cV$ in $\Sigma_{\Xi\mathsf{cd}}^N$ outputs $\top$ if $\cV_j$ outputs $\top$ for all $j\in[N]$, and outputs $\bot$ otherwise.
$\Sigma_{\Xi\mathsf{cd}}^{N}$ is a certified everlasting zero-knowledge proof for $A$ with $(1-\negl(\lambda))$-completeness and $\negl(\lambda)$-soundness.
\end{theorem}

\ifnum\submission=1
\begin{proof}
We provide a proof in \cref{Sec:sequential_everlasting}.
\end{proof}
\else
\input{Proof_sequential}
\fi

%% file: Proof_sequential.tex
\begin{proof}[Proof of \cref{thm:sequential_everlasting_zero_proof}]
It is easy to show that $\Sigma_{\Xi\mathsf{cd}}^N$ satisfies $(1-\negl(\lambda))$-completeness and $\negl(\lambda)$-soundness.
Moreover, as proven in \cite{GoldOre94}, the sequential repetition of a computational zero-knowledge proof 
preserves the computational zero-knowledge property.
Let us show that $\Sigma_{\Xi\mathsf{cd}}^N$ satisfies certified everlasting zero-knowledge.
For clarity, we describe how $\langle \cP(\witness^{\otimes Nk(|\statement|)}),\cV^*(\xi_1)\rangle (\statement)$ runs against any QPT verifier $\cV^*$ with an input $\xi_1$,
where $\witness$ is a witness and $\statement$ is the instance.
\begin{description}
\item[$\langle \cP(\witness^{\otimes Nk(|\statement|)}),\cV^*(\xi_1)\rangle (\statement)$:] $ $
\begin{enumerate}
    \item For $1\leq j\leq N$, $\cV^*$ and $\cP$ run 
    $\langle \cP_j(\witness^{\otimes k(|\statement|)}),\cV_j^{*}(\xi_j)\rangle(\statement)$ sequentially
    to get the outputs 
    \begin{align}
    \xi_{j+1}\seteq\mathsf{OUT_{\cV^*_j}}\langle \cP_j(\witness^{\otimes k(|\statement|)}),\cV_j^{*}(\xi_j)\rangle(\statement)
   \end{align} 
    and
    \begin{align}
    \mathsf{OUT_{\cP_j}}\langle \cP_j(\witness^{\otimes k(|\statement|)}),\cV_j^{*}(\xi_j)\rangle(\statement)=\top/\bot,
    \end{align}
    respectively.
    \item $\cV^*$ outputs $\xi_{N+1}$.
    \item $\cP$ outputs $\top$ if $\mathsf{OUT}_{\cP_j}\langle\cP_j(\witness^{\otimes k(|\statement|)}),\cV_j^*(\xi_j)\rangle=\top$ for all $j\in[N]$,
    and outputs $\bot$ otherwise.
\end{enumerate}
\end{description}

Since $\Sigma_{\Xi\mathsf{cd}}$ satisfies the certified everlasting zero-knowledge property,
for each $j\in[N]$ and any $\cV_j^*$ there exists a QPT algorithm (a simulator) $\cS_j(\statement,\cV_j^*,\cdot~)$ such that the following holds for any $\statement$ and $\witness$.
\begin{align}
    \mathsf{OUT'}_{\cP_j,\cV_j^*}\langle \cP_j(\witness^{\otimes k(|\statement|)}),\cV_j^*(\cdot)\rangle (\statement)\approx_s \cS_j(\statement,\cV^*_j,\cdot~).
\end{align}

We show that for any $\cV^*$ there exists a QPT algorithm (a simulator)
$\cS(\statement,\cV^*,\cdot~)$ such that the following holds for any $\statement$ and $\witness$.
\begin{align}
    \mathsf{OUT'}_{\cP,\cV^*}\langle \cP(\witness^{\otimes Nk(|\statement|)}),\cV^*(\cdot)\rangle (\statement)\approx_s
    \cS(\statement,\cV^*,\cdot~).
\end{align}
Let us define the simulator $\cS$ as follows.
\begin{description}
\item[The simulator $\cS(\statement,\cV^*,\xi_1)$:] $ $
\begin{enumerate}
     \item For $1\leq j\leq N$, $\cS$ runs $\cS_j(\statement,\cV^*_j,\cdot~)$ on $\xi_j$ to get $\cS_j(\statement,\cV^*_j,\xi_j)=(\bot,\bot)/(\top,\xi_{j+1})$ sequentially.
    If $\cS_j(\statement,\cV^*_j,\xi_j)=(\bot,\bot)$, then $\xi_{j+1}\seteq \bot$ for each $j\in[N]$.
    \item $\cS$ outputs $(\bot,\bot)$ if $\cS_j(\statement,\cV^*_j,\xi_j)=(\bot,\bot)$ for some $j\in[N]$, and outputs $(\top,\xi_{N+1})$ otherwise.
\end{enumerate}
\end{description}

We define the sequence of hybrids $\sfhyb{i}{}(\xi_1)$ as follows.
\begin{description}
\item[$\sfhyb{i}{}(\xi_1)$:] $ $
\begin{enumerate}
    \item For $1\leq j\leq i$, $\cV^*$ and $\cP$ run 
    $\langle \cP_j(\witness^{\otimes k(|\statement|)}),\cV_j^{*}(\xi_j)\rangle(\statement)$ sequentially to get the outputs
    \begin{align}
    \xi_{j+1}\seteq\mathsf{OUT_{\cV^*_j}}\langle \cP_j(\witness^{\otimes k(|\statement|)}),\cV_j^{*}(\xi_j)\rangle(\statement)
    \end{align}
    and 
    \begin{align}
    \mathsf{OUT}_{\cP_j}\langle \cP_j(\witness^{\otimes k(|\statement|)}),\cV_j^{*}(\xi_j)\rangle(\statement)=\top/\bot,
    \end{align}
    respectively.
    \item For $i+1\leq j\leq N$, $\cS$ runs $\cS_j(\statement,\cV^*_j,\cdot~)$ on $\xi_j$ to get
     $\cS_j(\statement,\cV^*_j,\xi_j)=(\bot,\bot)/(\top,\xi_{i+1})$ sequentially.
    If $\cS_j(\statement,\cV^*_j,\xi_{j})=(\bot,\bot)$, then $\xi_{j+1}\coloneqq \bot$ for each $j\in[N]$.
    \item The output of $\sfhyb{i}{}(\xi_1)$ is $(\bot,\bot)$ if $\mathsf{OUT}_{\cP_j}\langle \cP_j(\witness^{\otimes k(|\statement|)}),\cV_j^{*}(\xi_j)\rangle(\statement)=\bot$ for some $j\in[i]$ or $\cS_j(\statement,\cV^*_j,\xi_{j})=(\bot,\bot)$ for some $j\in\{i+1,\cdots N\}$.
    Otherwise, the output of $\sfhyb{i}{}(\xi_1)$ is $(\top,\xi_{N+1})$.
\end{enumerate}
\end{description}

$\sfhyb{0}{}(\cdot)$ and $\sfhyb{N}{}(\cdot)$ correspond to $\cS(\statement,\cV^*,\cdot~)$ and 
$\mathsf{OUT'}_{\cP,\cV^*}\langle \cP(\witness^{\otimes Nk(|\statement|)}),\cV^*(\cdot)\rangle (\statement)$, respectively.
Therefore, it suffices to prove that no distinguisher can distinguish $\sfhyb{i}{}(\cdot)$ from $\sfhyb{i+1}{}(\cdot)$ 
for any $i\in[N-1]$.
We assume that there exists a distinguisher $\cD'$ that distinguishes 
$\left(\sfhyb{i}{}(\cdot)\otimes I\right)\sigma$ from $\left(\sfhyb{i+1}{}(\cdot)\otimes I\right)\sigma$ for a 
certain state $\sigma$, and construct a distinguisher $\cD$
that breaks the certified everlasting zero-knowledge property of $\Sigma_{\Xi\mathsf{cd}}$.
$\cD$ can access to the channel $\mathsf{O}(\cdot)$, which is either $\cS_{i+1}(\statement,\cV^*_{i+1},\cdot~)$ 
or $\mathsf{OUT'}_{\cP_{i+1},\cV^*_{i+1}}\langle \cP_{i+1}(\witness^{\otimes k(|\statement|)}),\cV^*_{i+1}(\cdot)\rangle (\statement)$,
and guesses whether $\mathsf{O}(\cdot)$ is
$\cS_{i+1}(\statement,\cV^*_{i+1},\cdot~)$ or
$\mathsf{OUT'}_{\cP_{i+1},\cV^*_{i+1}}\langle \cP_{i+1}(\witness^{\otimes k(|\statement|)}),\cV^*_{i+1}(\cdot)\rangle (\statement)$.
Let us define $\cD$ as follows. 

\begin{description}
\item[The distinguisher $\cD(\xi_1)$:] $ $
\begin{enumerate}
    \item For $1\leq j\leq i$, $\cD$ 
    runs $\langle \cP_j(\witness^{\otimes k(|\statement|)}),\cV_j^{*}(\xi_j)\rangle(\statement)$ sequentially to get 
    $\xi_{j+1}\seteq\mathsf{OUT_{\cV^*_j}}\langle \cP_j(\witness^{\otimes k(|\statement|)}),\cV_j^{*}(\xi_j)\rangle(\statement)$ 
    and $\mathsf{OUT_{\cP_j}}\langle \cP_j(\witness^{\otimes k(|\statement|)}),\cV_j^{*}(\xi_j)\rangle(\statement)=\top/\bot$.
    \item 
    $\cD$ runs $\mathsf{O}(\xi_{i+1})$ to get $\mathsf{O}(\xi_{i+1})=(\bot,\bot)/(\top,\xi_{i+2})$.
    If $\mathsf{O}(\xi_{i+1})=(\bot,\bot)$, $\cD$ sets $\xi_{i+2}\coloneqq \bot$.
    \item For $i+2\leq j\leq N$, $\cD$ runs $\cS_j(\statement,\cV^*_j,\cdot~)$ on $\xi_j$ to get $\cS_{j}(\statement,\cV^{*}_{j},\xi_j)=(\bot,\bot)/(\top,\xi_{j+1})$ sequentially.
    If $\cS_{j}(\statement,\cV^{*}_{j},\xi_j)=(\bot,\bot)$, $\cD$ sets $\xi_{j+1}\coloneqq \bot$.
    \item 
    $\cD$ outputs $(\bot,\bot)$ if $\mathsf{OUT_{\cP_j}}\langle \cP_j(\witness^{\otimes k(|\statement|)}),\cV_j^{*}(\xi_j)\rangle(\statement)=\bot$ for some $j\in[i]$,
    $\mathsf{O}(\xi_{i+1})=(\bot,\bot)$
    or $\cS_j(\statement,\cV^*_j,\xi_j)=(\bot,\bot)$ for some $j\in\{i+2,\cdots,N\}$, and outputs $(\top,\xi_{N+1})$ otherwise.
    \item $\cD$ sends the output of $\cD$ to $\cD'$.
    \item If $\cD'$ outputs $b$, $\cD$ outputs $b$.
\end{enumerate}
\end{description}

We can see that $\cD$ generates $\left(\sfhyb{i}{}(\cdot)\otimes I\right)\sigma$ when $\mathsf{O}(\cdot)$ is $\cS_{i+1}(\statement,\cV^*_{i+1},\cdot~)$ and $\cD$ takes $\sigma$ as input.
Similarly, we can see that $\cD$ generates $\left(\sfhyb{i+1}{}(\cdot)\otimes I\right)\sigma$ when $\mathsf{O}(\cdot)$ is $\mathsf{OUT'}_{\cP_{i+1},\cV^*_{i+1}}\langle \cP_{i+1}(\witness^{\otimes k(|\statement|)}),\cV^*_{i+1}(\cdot)\rangle (\statement)$ and $\cD$ takes $\sigma$ as input.
Therefore, if $\cD'$ distinguishes $\left(\sfhyb{i}{}(\cdot)\otimes I\right)\sigma$ 
from $\left(\sfhyb{i+1}{}(\cdot)\otimes I\right)\sigma$,
then 
$\cD$ can distinguish $\cS_{i+1}(\statement,\cV^*_{i+1},\cdot~)$ from
$\mathsf{OUT'}_{\cP_{i+1},\cV^*_{i+1}}\langle \cP_{i+1}(\witness^{\otimes k(|\statement|)}),\cV^*_{i+1}(\cdot)\rangle (\statement)$.
This contradicts the certified everlasting zero-knowledge property of $\Sigma_{\Xi\mathsf{cd}}$, which completes the proof.
\end{proof}

%% file: reference.tex
\newcommand{\etalchar}[1]{$^{#1}$}

%% file: Sec_Appendix.tex
\section{Proof of \texorpdfstring{\cref{prop:abort}}{Proposition~\ref{prop:abort}}}\label{proof:abort}
\begin{proof}[Proof of \cref{prop:abort}]
We prove the proposition by contradiction.
Let $p$ be the probability that $\cS^{(1)}$ does not abort.
Assume that the probability $p$ satisfies $|p-\frac{1}{m}|\geq \frac{1}{q(\lambda)}$ for a polynomial $q$.
Then, we can construct an adversary $\cB$ that breaks the computational hiding of $\Sigma_{\mathsf{ccd}}$. 
Let us describe how $\cB$ works below.
\begin{enumerate}
    \item $\cB$ generates $c\leftarrow[m]$ and $x,z\lrun\{0,1\}^n$.
    \item $\cB$ sends $m_0\coloneqq \{x_i,z_i\}_{i\in S_c}$ and $m_1\coloneqq0^{10}$ to the challenger.
    $\cB$ receives commitments from the challenger which is either $\{\cment_i(x_i),\cment_i(z_i)\}_{i\in S_c}$ or $\{\cment_i(0),\cment_i(0)\}_{i\in S_c}$.
    \item $\cB$ generates $\{\cment_i(0),\cment_i(0)\}_{i\in \overline{S}_c}$.
    \item $\cB$ generates $X^xZ^z\sigma(c)Z^zX^x$.
    $\cB$ appends commitments and $\xi$ to the quantum state in the ascending order.
    If the commitments for $i\in S_c$ are $\{\cment_i(x_i),\cment_i(z_i)\}_{i\in S_c}$,
    $\cB$ obtains $X^xZ^z\sigma(c)Z^zX^x\otimes \cment(x^{S_c})\otimes\cment(z^{S_c})\otimes \xi$.
    If the commitments for $i\in S_c$ are $\{\cment_i(0),\cment_i(0)\}_{i\in S_c}$,
    $\cB$ obtains $X^xZ^z\sigma(c)Z^zX^x\otimes \cment(0^n)\otimes\cment(0^n)\otimes \xi$.
    \item $\cB$ runs $V_1^*$ on it to obtain $(c',\{\cert_{i,x}',\cert_{i,z}'\}_{i\in \overline{S}_{c'}})$. 
    $\cB$ outputs 0 when $c\neq c'$.
    $\cB$ outputs 1 when $c=c'$.
\end{enumerate}

When $\cB$ receives $\{\cment_i(x_i),\cment_i(z_i)\}_{i\in S_c}$ from the challenger,
it outputs 1 with probability $p$
since it simulates $\cS^{(1)}$.
When $\cB$ receives $\{\cment_i(0),\cment_i(0)\}_{i\in S_c}$ from the challenger, on the other hand,
it outputs 1 with probability $\frac{1}{m}$,
because $(X^xZ^z\sigma(c) Z^zX^x)\otimes \cment(0^n)\otimes\cment(0^n)\otimes \xi$ is independent of $c$.
(Note that $\sigma(c)$ is one-time padded by $x,z$.)
Therefore if there exists some polynomial $q$ such that $|p-\frac{1}{m}|\geq \frac{1}{q(\lambda)}$, 
$\cB$ can break the computational hiding of $\Sigma_{\mathsf{ccd}}$
from (the computational hiding version of) \cref{lemma:eachbit}.
\end{proof}

\section{Proof of \texorpdfstring{\cref{lem:computational_zero}}{Lemma~\ref{lem:computational_zero}}}\label{proof:computational_zero}
\begin{proof}[Proof of \cref{lem:computational_zero}]
This proof is similar to the proof of Lemma~\ref{lem:everlasting_zero}.
For a subset $S_c\subseteq[n]$ and $x,z\in\bit^n$, let us define $x^{S_c}\seteq (x^{S_c}_1,x^{S_c}_2,\cdots, x^{S_c}_n)$ and $z^{S_c}\seteq (z^{S_c}_1,z^{S_c}_2,\cdots,z^{S_c}_n)$,
where $x^{S_c}_{i}=x_{i}$ and $z^{S_c}_{i}=z_{i}$ for $i\in S_{c}$,
and $x^{S_c}_{i}=z^{S_c}_{i}=0$ for $i\notin S_{c}$. 

For clarity, we describe how the interactive algorithm $\langle \cP(\witness^{\otimes k(|\statement|)}),\cV^*(\xi)\rangle (\statement)$ runs against a QPT verifier $\cV^*$ with an input $\xi$,
where $\witness$ is the witness and $\statement$ is the instance.
\begin{description}
\item[$\langle \cP(\witness^{\otimes k(|\statement|)}),\cV^*(\xi)\rangle (\statement)$:] $ $
\begin{enumerate}
    \item $\cP$ generates $x,z\leftarrow \{0,1\}^n$, and computes 
    \begin{align}
    &(\cment_i(x_i),\dment_i(x_i),\ck_i(x_i))\leftarrow\Commit(1^\lambda,x_i) \\ 
    &(\cment_i(z_i),\dment_i(z_i),\ck_i(z_i))\leftarrow\Commit(1^\lambda,z_i)
    \end{align}
    for all $i\in [n]$.
    $\cP$ sends $\msg_1\seteq(X^xZ^z\rho_\hist Z^zX^x)\otimes \cment(x)\otimes\cment(z)$ to $\cV^*$. 
    \item $\cV^*$ appends $\xi$ to the received state, and
   runs a QPT circuit $V_1^*$ on it to obtain $(c,\{\cert_{i,x}',\cert_{i,z}'\}_{i\in\overline{S}_c})$.
   $\cV^*$ sends $\msg_2\seteq(c,\{\cert_{i,x}',\cert_{i,z}'\}_{i\in\overline{S}_c})$ to $\cP$.
    \item $\cP$ sends $\msg_3\seteq\{\dment_i(x_i),\dment_i(z_i)\}_{i\in S_c}$ to $\cV^*$.
    \item $\cV^*$ appends $\msg_3$ to its state, and runs a QPT circuit $V_2^*$ on it. $\cV^*$ outputs its state $\xi'$.
    \end{enumerate}
\end{description}

Next let us define a simulator $\cS^{(1)}$ as follows.

\begin{description}
\item [The simulator $\cS^{(1)}(\statement,\cV^*,\xi)$:] $ $
\begin{enumerate}
    \item Pick $c\leftarrow [m]$ and $x,z\leftarrow \{0,1\}^n$. 
    Compute 
    \begin{align}
   & (\cment_i(x^{S_c}_i),\dment_i(x^{S_c}_i),\ck_i(x^{S_c}_i))\leftarrow\Commit(1^\lambda,x^{S_c}_i)\\
   & (\cment_i(z^{S_c}_i),\dment_i(z^{S_c}_i),\ck_i(z^{S_c}_i))\leftarrow\Commit(1^\lambda,z^{S_c}_i)
    \end{align} 
    for all $i\in [n]$.
    \item Generate $(X^xZ^z\sigma(c) Z^zX^x) \otimes \cment(x^{S_c})\otimes \cment(z^{S_c})\otimes \xi$, where $\sigma(c)\seteq \rho_\simulator^{\statement,S_c}\otimes\left(\prod_{i\in \overline{S}_c}|0\rangle\langle0|_i\right)$.
    Run $V_1^*$ on the state to obtain $(c',\{\cert_{i,x}',\cert_{i,z}'\}_{i\in\overline{S}_{c'}})$.
    \item If $c'\neq c$, abort and output a fixed state $\eta$ and the flag state $\mathsf{fail}$.
    \item Append $\{\dment_i(x^{S_c}_i),\dment_i(z^{S_c}_i)\}_{i\in S_c}$ to its quantum state, and run $V_2^*$ on the state.
    $\cS$ outputs the output state and the flag state $\mathsf{success}$.
\end{enumerate}
\end{description}

Let us also define other two simulators $\cS^{(2)}$ and $\cS^{(3)}$ as follows.
\begin{description}
\item [The modified simulator $\cS^{(2)}(\statement,\witness^{\otimes k(|\statement|)},\cV^*,\xi)$:] $ $
   It is the same as $\cS^{(1)}$ except that
   $\sigma(c)$ is replaced with $\rho_{\hist}$. 
\end{description}

\begin{description}
\item [The simulator $\cS^{(3)}(\statement,\witness^{\otimes k(|\statement|)},\cV^*,\xi)$:] $ $
$\cS^{(3)}(\statement,\witness^{\otimes k(|\statement|)},\cV^*,\cdot~)$ is the channel that postselects the output of
\\$\cS^{(2)}(\statement,\witness^{\otimes k(|\statement|)},\cV^*,\cdot~)$ on the non-aborting state. More precisely, if we
write $\cS^{(2)}(\statement,\witness^{\otimes k(|\statement|)},\cV^*,\rho_{in})=p \rho_{out}\otimes \mathsf{success}+(1-p)\eta\otimes \mathsf{fail}$, 
where $p$ is the non-aborting probability,
$\cS^{(3)}(\statement,\witness^{\otimes k(|\statement|)},\cV^*,\rho_{in})= \rho_{out}$.
\end{description}

\cref{lem:computational_zero} is shown from the following \cref{prop:abort_comp,prop:indistinguishable_S_S'_comp,prop:indistinguishable_S'_V^*_comp} (whose proofs will be given later)
and quantum rewinding lemma(\cref{lemma:rewinding}), which is used to reduce the probability that $\cS^{(1)}$ aborts to $\negl(\lambda)$.
In fact, from \cref{prop:abort_comp,lemma:rewinding}, there exists a quantum circuit $\cS^{(0)}$ of size at most 
$O(m\,{\rm poly}(n){\rm size}(\cS^{(1)}))$
such that the probability $\cS^{(0)}$ aborts is $\negl(\lambda)$,
and the output quantum states of $\cS^{(0)}$ and $\cS^{(1)}$ are $\negl(\lambda)$-close when they do not abort.
From \cref{prop:indistinguishable_S_S'_comp,prop:indistinguishable_S'_V^*_comp},
$\cS^{(0)}$ is $\negl(\lambda)$-close to the run of the real protocol, which completes the proof.
\end{proof}

\begin{proposition}\label{prop:abort_comp}
If $\Sigma_{\mathsf{ccd}}$ is computational hiding, then the probability that $\cS^{(1)}$ does not abort is $\frac{1}{m}\pm\negl(\lambda)$.
\end{proposition}

\begin{proposition}\label{prop:indistinguishable_S_S'_comp}
$\cS^{(1)}(\statement,\cV^*,\cdot~)\approx_{s}\cS^{(2)}(\statement,\witness^{\otimes k(|\statement|)},\cV^*,\cdot~)$
for any $\statement\in A_{\yes}\cap\bit^\lambda$ and any $\witness\in R_A(\statement)$.
\end{proposition}

\begin{proposition}\label{prop:indistinguishable_S'_V^*_comp}
If $\Sigma_{\mathsf{ccd}}$ is computational hiding,
$\cS^{(3)}(\statement,\witness^{\otimes k(|\statement|)},\cV^*,\cdot~)\approx_{c}\mathsf{OUT}_{\cV^*}\langle \cP(\witness^{\otimes k(|\statement|)}),\cV^*(\cdot)\rangle (\statement)$.
\end{proposition}

\begin{proof}[Proof of \cref{prop:abort_comp}]
This proof is the same as the proof of \cref{prop:abort}.
\end{proof}

\begin{proof}[Proof of \cref{prop:indistinguishable_S_S'_comp}]
This proof is the same as the proof of \cref{prop:indistinguishable_S_S'_ever}.
\end{proof}

\begin{proof}[Proof of \cref{prop:indistinguishable_S'_V^*_comp}]
We prove the proposition by contradiction.
We construct an adversary $\cB$ that breaks the security of the computationally hiding of $\Sigma_{\mathsf{ccd}}$ by assuming 
the existence of a distinguisher $\cD$ that distinguishes two states $\delta_0$ and $\delta_1$,
\begin{align}
&\delta_0\seteq(\mathsf{OUT}_{\cV^*}\langle \cP(\witness^{\otimes k(|\statement|)}),\cV^*(\cdot)\rangle (\statement)\otimes I)\sigma\\
&\delta_1\seteq(\cS^{(3)}(\statement,\witness^{\otimes k(|\statement|)},\cV^*,\cdot~)\otimes I)\sigma,
\end{align}
with a certain state $\sigma$.
Let us describe how $\cB$ works.
\begin{enumerate}
    \item $\cB$ generates $c\leftarrow[m]$ and $x,z\lrun\{0,1\}^n$.
    \item $\cB$ sends $m_0\coloneqq \{x_i,z_i\}_{i\in\overline{S}_c}$ and $m_1\coloneqq0^{2n-10}$ to the challenger.
    $\cB$ receives commitments from the challenger which is either $\{\cment_i(x_i),\cment_i(z_i)\}_{i\in \overline{S}_c}$ or $\{\cment_i(0),\cment_i(0)\}_{i\in \overline{S}_c}$.
    \item $\cB$ computes 
    \begin{align}
 &   (\cment_i(x_i),\dment_i(x_i),\ck_i(x_i))\leftarrow\Commit(1^\lambda,x_i)\\
 &   (\cment_i(z_i),\dment_i(z_i),\ck_i(z_i))\leftarrow\Commit(1^\lambda,z_i)
    \end{align}
    for $i\in S_c$ by itself.
    \item $\cB$ generates $X^xZ^z\rho_{\hist}Z^zX^x$.  
    $\cB$ appends commitments and $\sigma$ to the quantum state.
    If the commitments for ${i\in \overline{S}_c}$ are $\{\cment_i(x_i),\cment_i(z_i)\}_{i\in \overline{S}_c}$,
    $\cB$ obtains $X^xZ^z\rho_{\hist}Z^zX^x\otimes \cment(x)\otimes\cment(z)\otimes \sigma$.
    If the commitments for ${i\in \overline{S}_c}$ are $\{\cment_i(0),\cment_i(0)\}_{i\in \overline{S}_c}$,
    $\cB$ obtains $X^xZ^z\rho_{\hist}Z^zX^x\otimes \cment(x^{S_c})\otimes\cment(z^{S_c})\otimes \sigma$.
    \item $\cB$ runs $V_1^*$ on it to obtain $(c',\{\cert_{i,x}',\cert_{i,z}'\}_{i\in \overline{S}_{c'}})$. $\cB$ aborts when $c\neq c'$.
    \item $\cB$ appends $\{\dment_i(x_i),\dment_{i}(z_i)\}_{i\in S_c}$ to the post-measurement state and runs $V_2^*$ on it.
    \item $\cB$ passes the output state to $\cD$.
    \item When $\cD$ outputs $b$, $\cB$ outputs $b$.
\end{enumerate}

When $\cB$ receives $\{\cment_i(x_i),\cment_i(z_i)\}_{i\in \overline{S}_c}$ from the challenger and it does not abort,
it simulates \\$\mathsf{OUT}_{\cV^*}\langle \cP(\witness^{\otimes k(|\statement|)}),\cV^*(\cdot)\rangle(\statement)$.
Because $(X^xZ^z\rho_{\hist}Z^zX^x)\otimes \cment(x)\otimes\cment(z)\otimes \sigma$ is independent of $c$,
the probability that $\cB$ does not abort is $\frac{1}{m}$.
Therefore, $\cB$ can simulate ${\mathsf{OUT}}_{\cV^*}\langle \cP(\witness^{\otimes k(|\statement|)}),\cV^*(\cdot)\rangle(\statement)$
with probability $\frac{1}{m}$.

 When $\cB$ receives $\{\cment_i(0),\cment_i(0)\}_{i\in\overline{S}_c}$ from the challenger and it does not abort,
it simulates $\cS^{(3)}(\statement,\witness^{\otimes k(|\statement|)},\cV^*,\cdot~)$.
The probability that $\cB$ does not abort is $\frac{1}{m}\pm\negl(\lambda)$ from \cref{prop:abort_comp,prop:indistinguishable_S_S'_comp}.
Therefore, $\cB$ can simulate 
$\cS^{(3)}(\statement,\witness^{\otimes k(|\statement|)},\cV^*,\cdot~)$ with probability
$\frac{1}{m}\pm\negl(\lambda)$.

Therefore, if there exists the distinguisher $\cD$ that distinguishes 
$\delta_0$ and $\delta_1$,
$\cB$ can distinguish $\{\cment_i(x_i),\cment_i(z_i)\}_{i\in \overline{S}_c}$ from $\{\cment_i(0),\cment_i(0)\}_{i\in \overline{S}_c}$.
From (the computational hiding version of) \cref{lemma:eachbit},
this contradicts the computational hiding of $\Sigma_{\mathsf{ccd}}$.

\end{proof}

%% file: App_sum_binding.tex

\section{Commitment with Certified Everlasting Hiding and Sum-Binding}\label{app:sum_binding}
In this appendix, we define and construct commitment with certified
everlasting hiding and statistical sum-binding.

\subsection{Definition}
\begin{definition}[Commitment with Certified Everlasting Hiding and Sum-Binding (Syntax)]\label{def:commitment_with_cd_sumbinding}
Let $\lambda$ be the security parameter, and let $p$, $q$, $r$ and $s$ be some polynomials.
Commitment with certified everlasting hiding and sum-binding consists of a tuple of algorithms $(\Commit,\Verify,\Delete,\Cert)$ 
with message space $\Ms:=\{0,1\}$, commitment space $\Cs:=\cQ^{\otimes p(\lambda)}$,
decommitment space $\mathcal{D}:=\{0,1\}^{q(\lambda)}$, key space $\Ks:=\{0,1\}^{r(\lambda)}$
and deletion certificate space $\mathcal{E}:= \{0,1\}^{s(\lambda)}$.

\begin{description}
\item[$\Commit (1^{\lambda},b)\rightarrow (\cment,\dment,\ck)$:] 
The commitment algorithm takes as input a security parameter $1^{\lambda}$ and a message $b\in\{0,1\}$, and outputs a commitment $\cment\in\Cs$, a decommitment $\dment\in\mathcal{D}$, and a key $\ck\in\Ks$.
\item[$\Verify(\cment,\dment,b)\rightarrow \top~or~\bot $:]
The verification algorithm takes as input $\cment$, $\dment$ and $b$,
and outputs $\top$ or $\bot$.
\item[$\Delete(\cment)\rightarrow \cert$:]The deletion algorithm takes $\cment$ as input, and outputs a certificate $\cert\in \mathcal{E}$. 
\item[$\Cert(\cert,\ck)\rightarrow \top~or~\bot$:]The certification algorithm takes $\cert$ and $\ck$ as input,
and outputs $\top$ or $\bot$.
\end{description}
\end{definition}

\begin{definition}[Correctness]\label{def:correctness_bit_commitment} 
There are two types of correctness, namely, decommitment correctness and deletion correctness.
\begin{description}
\item[Decommitment correctness:]
There exists a negligible function $\negl$ such that for any $\lambda\in\mathbb{N}$ and $b\in\{0,1\}$,
\begin{align}
    \Pr[\Verify (\cment,\dment,b)=\top \mid (\cment,\dment,\ck)\leftarrow\Commit(1^{\lambda},b)]\geq1-\negl(\lambda).
\end{align}

\item[Deletion correctness:]
There exists a negligible function $\negl$ such that for any $\lambda \in\mathbb{N}$ and $b\in\{0,1\}$, 
\begin{align}
    \Pr[\Cert(\cert,\ck)=\top \mid (\cment,\dment,\ck)\leftarrow\Commit(1^\lambda,b),\cert\leftarrow\Delete(\cment)]\geq1-\negl(\lambda).
\end{align}
\end{description}
\end{definition}

\begin{definition}[$\epsilon$-Sum-Binding]\label{def:sum_binding}
For any $\cment$, $\dment$, and $\dment'$, it holds that 
\begin{align}
    \Pr[\Verify(\cment,\dment,0)=\top]+\Pr[\Verify(\cment,\dment',1)=\top]\leq 1+\epsilon.
\end{align}
We call $\epsilon$-sum-binding just sum-binding if $\epsilon$ is negligible.
\end{definition}

\begin{definition}[Computational Hiding]\label{def:computationally_hiding}
Let $\Sigma\coloneqq(\Commit,\Verify,\Delete,\Cert)$. 
Let us consider the following security experiment $\expa{\Sigma,\cA}{c}{hide}(\lambda,b)$ against
any QPT adversary $\mathcal{A}$.
\begin{enumerate}
\item The challenger computes $(\cment,\dment,\ck)\leftarrow \Commit(1^{\lambda},b)$, and sends $\cment$ to $\mathcal{A}$.
\item $\cA$ outputs $b'\in\{0,1\}$.
\item The output of the experiment is $b'$. 
\end{enumerate}
Computational hiding means that the following is satisfied for any QPT $\cA$.
\begin{align}
\advb{\Sigma,\cA}{c}{hide}(\lambda)\seteq
\left|\Pr[ \expa{\Sigma,\cA}{c}{hide}(\lambda,0)=1]-
\Pr[\expa{\Sigma,\cA}{c}{hide}(\lambda,1)=1]\right|\leq \negl(\lambda).
\end{align}
\end{definition}

\begin{definition}[Certified Everlasting Hiding]\label{def:everlasting_hiding}
Let $\Sigma\coloneqq(\Commit,\Verify,\Delete,\Cert)$. 
Let us consider the following security experiment $\expa{\Sigma,\cA}{ever}{hide}(\lambda,b)$ against
$\mathcal{A}=(\mathcal{A}_1,\mathcal{A}_2)$ consisting of any QPT adversary $\mathcal{A}_1$ and any unbounded adversary $\mathcal{A}_2$.

\begin{enumerate}
\item The challenger computes $(\cment,\dment,\ck)\leftarrow \Commit(1^{\lambda},b)$, and sends $\cment$ to $\mathcal{A}_1$.
\item At some point, $\mathcal{A}_1$ sends $\cert$ to the challenger, and sends its internal state to $\mathcal{A}_2$.
\item The challenger computes $\Cert(\cert,\ck)$. If the output is $\top$, then the challenger outputs $\top$,
and sends $(\dment,\ck)$ to $\cA_2$.
Else, the challenger outputs $\bot$, and sends $\bot$ to $\cA_2$.
\item $\cA_2$ outputs $b'\in\{0,1\}$.
\item If the challenger outputs $\top$, then the output of the experiment is $b'$. Otherwise, the output of the experiment is $\bot$.

\end{enumerate}

We say that it is certified everlasting hiding if the following is satisfied for any $\cA=(\cA_1,\cA_2)$. 
\begin{align}
\advb{\Sigma,\cA}{ever}{hide}(\lambda)\coloneqq
\left|\Pr[ \expa{\Sigma,\cA}{ever}{hide}(\lambda,0)=1]-
\Pr[\expa{\Sigma,\cA}{ever}{hide}(\lambda,1)=1]\right|\leq \negl(\lambda).
\end{align}
\end{definition}

\subsection{Construction}

Though the construction is essentially the same as that in \cref{sec:bit_comm_construction}, we give the full description for clarity.
Let $\lambda$ be the security parameter, and let $p$, $q$, $r$, $s$, $t$ and $u$ be some polynomials.
We construct a bit commitment with certified everlasting hiding and sum-binding,
$\Sigma_{\mathsf{ccd}}=(\Commit,\Verify,\Delete,\Cert)$, with message space $\Ms=\{0,1\}$, 
commitment space $\Cs=\cQ^{\otimes p(\lambda)}\times \{0,1\}^{q(\lambda)}\times \{0,1\}^{r(\lambda)}$,
decommitment space $\mathcal{D}=\{0,1\}^{s(\lambda)}\times \{0,1\}^{t(\lambda)}$,
key space $\Ks=\{0,1\}^{r(\lambda)}$ 
and deletion certificate space $\mathcal{E}=\{0,1\}^{u(\lambda)}$
from the following primitives:
\begin{itemize}
\item Secret-key encryption with certified deletion 
$\Sigma_{\mathsf{skcd}}=\mathsf{SKE}.(\keygen,\Enc,\Dec,\Delete,\Verify)$, with plaintext space $\Ms=\{0,1\}$, ciphertext space $\Cs=\cQ^{\otimes p(\lambda)}$, key space $\Ks=\{0,1\}^{r(\lambda)}$, and deletion certificate space $\mathcal{E}=\{0,1\}^{u(\lambda)}$.
\item
Classical non-interactive commitment, $\Sigma_{\mathsf{com}}=\algo{Classical}.\Commit$,
with plaintext space $\{0,1\}^{s(\lambda)}$, randomness space $\{0,1\}^{t(\lambda)}$, and commitment space $\{0,1\}^{q(\lambda)}$.
\item
A hash function $H$ from $\{0,1\}^{s(\lambda)}$ to $\{0,1\}^{r(\lambda)}$ modeled as a quantumly-accessible random oracle.
\end{itemize}
The construction is as follows.
\begin{description}
\item[$\Commit(1^{\lambda},b)$:] $ $
\begin{itemize}
    \item Generate $\mathsf{ske.sk}\leftarrow\mathsf{SKE}.\keygen(1^{\lambda})$, $R\leftarrow\{0,1\}^{s(\lambda)}$, $R'\leftarrow\{0,1\}^{t(\lambda)}$, and a hash function $H$ from $\{0,1\}^{s(\lambda)}$ to $\{0,1\}^{r(\lambda)}$.
    \item Compute $\ske.\ct\leftarrow \SKE.\Enc(\ske.\sk,b)$, $f\leftarrow \algo{Classical}.\Commit(R;R')$, and $h\seteq H(R)\oplus\ske.\sk$.
    \item Output $\cment\seteq (\ske.\ct,f,h)$, $\dment\seteq (R,R')$, and $\ck\coloneqq\ske.\sk$.
\end{itemize}
\item[$\Verify(\cment,\dment,b)$:] $ $
\begin{itemize}
    \item Parse $\cment=(\ske.\ct,f,h)$ and $\dment=(R,R')$.
    \item Compute $\ske.\sk'\coloneqq H(R)\oplus h$.
    \item Compute $b'\leftarrow\SKE.\Dec(\ske.\sk',\ske.\ct)$.
    \item Output $\top$ if $f=\Classical.\Commit(R;R')$ and $b'=b$, and output $\bot$ otherwise.
\end{itemize}
\item[$\Delete(\cment)$:] $ $
\begin{itemize}
    \item Parse $\cment=(\ske.\ct,f,h)$. 
    \item Compute $\ske.\cert\leftarrow \SKE.\Delete(\ske.\ct)$.
    \item Output $\cert\coloneqq\ske.\cert$.
\end{itemize}
\item[$\Cert(\cert,\ck)$:] $ $
\begin{itemize}
    \item Parse $\cert=\ske.\cert$ and $\ck= \ske.\sk$.
    \item Output $\top/\bot\leftarrow \SKE.\Verify(\ske.\sk,\ske.\cert)$.
\end{itemize}
\end{description}

\paragraph{Correctness.}
The decommitment and deletion correctness easily follow from the correctness of $\Sigma_{\mathsf{skcd}}$.

\paragraph{Security.}
We prove the following three theorems.

\begin{theorem}\label{thm:sum_binding}
If $\Sigma_{\mathsf{com}}$ is perfect binding, then $\Sigma_{\mathsf{ccd}}$ is sum-binding.
\end{theorem}

\begin{theorem}\label{thm:everlasting_hiding}
If $\Sigma_{\mathsf{com}}$ is unpredictable and $\Sigma_\mathsf{skcd}$ is OT-CD secure, then $\Sigma_{\mathsf{ccd}}$ is certified everlasting hiding.
\end{theorem}

\begin{theorem}\label{thm:computationally_hiding}
If $\Sigma_{\mathsf{com}}$ is unpredictable and $\Sigma_{\mathsf{skcd}}$ is OT-CD secure, then $\Sigma_{\mathsf{ccd}}$ is computationally hiding.
\end{theorem}

\begin{proof}[Proof of \cref{thm:sum_binding}]
What we have to prove is that for any $\cment$, $\dment$, and $\dment'$, it holds that 
\begin{align}
    \Pr[\Verify(\cment,\dment,0)=\top]+\Pr[\Verify(\cment,\dment',1)=\top]\leq 1+\negl(\lambda).
\end{align}
Let $\dment=(R_0,R_0')$, $\dment'=(R_1,R_1')$,
and $\cment=(\ske.\ct,f,h)$.
Then,
\begin{align}
&    \Pr[\Verify(\cment,\dment,0)=\top]+\Pr[\Verify(\cment,\dment',1)=\top]\\ 
&    = \Pr[0\la \SKE.\Dec(h\oplus H(R_0),\ske.\CT)\wedge f=\Classical.\Commit(R_0;R_0')]\\
&    + \Pr[1\la \SKE.\Dec(h\oplus H(R_1),\ske.\CT)\wedge f=\Classical.\Commit(R_1;R_1')]\\
&    \le \Pr[0\la \SKE.\Dec(h\oplus H(\tilde{R}),\ske.\CT)\wedge f=\Classical.\Commit(\tilde{R};R_0')]\\
&    + \Pr[1\la \SKE.\Dec(h\oplus H(\tilde{R}),\ske.\CT)\wedge f=\Classical.\Commit(\tilde{R};R_1')]\\
&    \le \Pr[0\la \SKE.\Dec(h\oplus H(\tilde{R}),\ske.\CT)]
    + \Pr[1\la \SKE.\Dec(h\oplus H(\tilde{R}),\ske.\CT)]\\
&    = \Pr[0\la \SKE.\Dec(h\oplus H(\tilde{R}),\ske.\CT) \vee
    1\la \SKE.\Dec(h\oplus H(\tilde{R}),\ske.\CT)]\\
&\le1,
\end{align} 
where we have used perfect binding of $\Sigma_{\mathsf{com}}$ in the second inequality. 

\end{proof}

\begin{proof}[Proof of Theorem~\ref{thm:everlasting_hiding}]
It is the same as that of \cref{thm:everlasting_hiding_2}.
\end{proof}

\begin{proof}[Proof of Theorem~\ref{thm:computationally_hiding}]
It is the same as that of \cref{thm:computationally_hiding_2}.
\end{proof}

%% file: proof_of_bit.tex
\section{Proof of \texorpdfstring{\cref{lemma:eachbit}}{Lemma~\ref{lemma:eachbit}}}\label{sec:proof_of_bit}
Let us consider the following hybrids for $j\in\{0,1,...,n\}$.
\begin{description}
\item[$\sfhyb{j}{}$:] $ $
\begin{enumerate}
\item $\cA_1$ generates $(m^0,m^1)\in \bit^n\times\bit^n$ and sends it to the challenger.
\item The challenger computes 
\begin{align}
(\cment_i(m_i^1),\dment_i(m_i^1),\ck_i(m_i^1))\leftarrow \Commit(1^{\lambda},m^1_i) 
\end{align}
for $i\in[j]$ and
\begin{align}
(\cment_i(m_i^0),\dment_i(m_i^0),\ck_i(m_i^0))\leftarrow \Commit(1^{\lambda},m^0_i)
\end{align}
for each $i\in\{j+1,...,n\}$,
and sends $\{\cment_i(m_i^1)\}_{i\in[j]}$ and $\{\cment_i(m_i^0)\}_{i\in\{j+1,...,n\}}$ to $\mathcal{A}_1$.
Here, $m_i^b$ is the $i$-th bit of $m^b$.
\item At some point, $\mathcal{A}_1$ sends $\{\cert_i\}_{i\in[n]}$ to the challenger, 
and sends its internal state to $\mathcal{A}_2$.
\item The challenger computes $\Cert(\cert_i,\ck_i(m_i^1))$ for each $i\in[j]$
and
$\Cert(\cert_i,\ck_i(m_i^0))$ for each $i\in\{j+1,...,n\}$.
If the outputs are $\top$ for all $i\in[n]$, then the challenger outputs $\top$,
and sends 
$\{\dment_i(m_i^1),\ck_i(m_i^1)\}_{i\in[j]}$ 
and
$\{\dment_i(m_i^0),\ck_i(m_i^0)\}_{i\in\{j+1,...,n\}}$ 
to $\cA_2$.
Else, the challenger outputs $\bot$, and sends $\bot$ to $\cA_2$.
\item $\cA_2$ outputs $b'\in\{0,1\}$.
\item If the challenger outputs $\top$, then the output of the experiment is $b'$. Otherwise, the output of the experiment is $\bot$.
\end{enumerate}
\end{description}
It is clear that
$\sfhyb{0}{}=\expb{\Sigma,\cA}{bit}{ever}{hide}(\lambda,0)$ and $\sfhyb{n}{}=\expb{\Sigma,\cA}{bit}{ever}{hide}(\lambda,1)$. 
Furthermore, we can show  
\begin{align}
\left|\Pr[\sfhyb{j}{}=1]-\Pr[\sfhyb{j+1}{}=1]\right|\le\negl(\lambda)
\end{align}
for each $j\in\{0,1,...,n-1\}$. (Its proof is given below.)
From these facts, we obtain \cref{lemma:eachbit}.

Let us show the remaining one. To show it, let us assume that
$\left|\Pr[\sfhyb{j}{}=1]-\Pr[\sfhyb{j+1}{}=1]\right|$ is non-negligible.
Then, we can construct an adversary $\cB$ that can break the certified everlasting hiding of
$\Sigma_{\mathsf{ccd}}$ as follows.
\begin{enumerate}
\item
$\cB$ receives $(m^0,m^1)$ from $\cA_1$, and computes
\begin{align}
(\cment_i(m_i^1),\dment_i(m_i^1),\ck_i(m_i^1))\leftarrow \Commit(1^{\lambda},m^1_i)
\end{align}
for $i\in[j]$
and
\begin{align}
(\cment_i(m_i^0),\dment_i(m_i^0),\ck_i(m_i^0))\leftarrow \Commit(1^{\lambda},m^0_i)
\end{align}
for $i\in\{j+2,...,n\}$.
\item
$\cB$ sends $(m_{j+1}^0,m_{j+1}^1)$ to the challenger of 
$\expa{\Sigma_{\mathsf{ccd}},\cB}{ever}{hide}(\lambda,b')$,
and receives $\cment_{j+1}(m_{j+1}^{b'})$ from the challenger.
\item
$\cB$ sends 
$\{\cment_i(m_i^1)\}_{i\in[j]}$,
$\cment_{j+1}(m_{j+1}^{b'})$, and
$\{\cment_i(m_i^0)\}_{i\in\{j+2,...,n\}}$,
to $\cA_1$.
\item
$\cA_1$ sends $\{\cert_i\}_{i\in[n]}$ to $\cB$, and sends its internal state to $\cA_2$.
\item
$\cB$ sends $\cert_{j+1}$ to the challenger of
$\expa{\Sigma_{\mathsf{ccd}},\cB}{ever}{hide}(\lambda,b')$,
and receives 
$(\dment_{j+1}(m_{j+1}^{b'}),\ck_{j+1}(m_{j+1}^{b'}))$ or $\bot$ from the challenger.
If $\cB$ receives $\bot$ from the challenger, it outputs $\bot$ and aborts.
\item
$\cB$ sends all $\dment_i$ and $\ck_i$ to $\cA_2$.
\item
$\cA_2$ outputs $b''$.
\item
$\cB$ computes $\Cert$ for all $\cert_i$, and outputs $b''$ if all results are $\top$.
Otherwise, $\cB$ outputs $\bot$.
\end{enumerate}
It is clear that $\Pr[\cB\to1\mid b'=0]=\Pr[\sfhyb{j}{}=1]$ and
$\Pr[\cB\to1\mid b'=1]=\Pr[\sfhyb{j+1}{}=1]$.
By assumption,
$|\Pr[\sfhyb{j}{}=1]-\Pr[\sfhyb{j+1}{}=1]|$
is non-negligible, and therefore
$|\Pr[\cB\to1\mid b'=0]-\Pr[\cB\to1\mid b'=1]|$ is non-negligible,
which contradict the certified everlasting hiding of $\Sigma_{\mathsf{ccd}}$.